\documentclass[a4paper, 11pt]{article}
\usepackage[utf8]{inputenc}
\usepackage{graphicx} 
\usepackage{amsmath, amsfonts, amsthm, amsmath, bm, mathrsfs, amssymb}
\usepackage{multicol}
\usepackage{color}
\usepackage{algorithm}
\usepackage{algorithmic}
\usepackage{natbib}
\usepackage{booktabs} 

\usepackage[margin=15mm]{geometry} 

\usepackage{authblk} 

\newtheorem{Theo}{Theorem}
\newtheorem{Lemm}[Theo]{Lemma}

\newcommand{\argmin}{\mathop{\rm argmin}\limits}
\newcommand{\bs}{\boldsymbol}

\newcommand{\mac}{\mathcal}
\newcommand{\mbb}{\mathbb}

\newcommand{\mrm}{\mathrm}

\newcommand{\ra}{\rightarrow}

\newcommand{\lan}{\langle}
\newcommand{\ran}{\rangle}

\def\ba#1\ea{\begin{align*}#1\end{align*}} 
	\def\banum#1\eanum{\begin{align}#1\end{align}} 

\newif\ifdraft
\draftfalse 

\newcommand{\mymemo}[1]{\ifdraft \textcolor{blue}{#1}\fi}

\allowdisplaybreaks

\begin{document}

\title{$K$-means clustering for sparsely observed longitudinal data }
\author[1, 3, 4]{Michio Yamamoto} 
\author[2, 3]{Yoshikazu Terada}
\affil[1]{\small Graduate School of Human Sciences, Osaka University}
\affil[2]{\small Graduate School of Engineering Science, Osaka University}
\affil[3]{\small RIKEN AIP}
\affil[4]{\small Data Science and AI Innovation Research Promotion Center, Shiga University}
\maketitle

\begin{abstract}
    In longitudinal data analysis, observation points of repeated measurements over time often vary among subjects except in well-designed experimental studies.
    Additionally, measurements for each subject are typically obtained at only a few time points. 
    From such sparsely observed data, 
    identifying underlying cluster structures can be challenging.
    This paper proposes a fast and simple clustering method that generalizes the classical $k$-means method to identify cluster centers in sparsely observed data.
    The proposed method employs the basis function expansion to model the cluster centers, providing an effective way to estimate cluster centers from fragmented data.
    We establish the statistical consistency of the proposed method, 
    as with the classical $k$-means method.
    Through numerical experiments, we demonstrate that the proposed method performs competitively with, or even outperforms, existing clustering methods.
    Moreover,
    the proposed method offers significant gains in computational efficiency due to its simplicity.
    Applying the proposed method to real-world data illustrates its effectiveness in identifying cluster structures in sparsely observed data.
\end{abstract}

\vspace{25pt}


\section{Introduction}
\label{sec:Introduction}

Clustering subjects using longitudinal data is one of the basic analytical tasks in many scientific fields.
The situation in which all subjects in longitudinal data have measurements taken at the same times and in equal numbers is referred to as the \textit{regular} case \citep{yassouridis2018generalization} or a \textit{balanced} design \citep{smith2015essential}.
In this situation, Gaussian mixture models \citep{banfield1993model, mclachlan2019finite} or the classical $k$-means algorithm \citep{hartigan1979algorithm} can be applied.
However, these classical methods cannot be directly applied to the case of the so-called \textit{irregular} data \citep{yassouridis2018generalization} or \textit{unbalanced} design \citep{smith2015essential}, where the time and number of measurements differ for each subject. 
Furthermore, while measurement values change smoothly over time, these naive analyses ignore the dependency among time points, making them inappropriate for longitudinal data.

Functional data analysis (FDA), which assumes a continuous function underlying the discrete measurements, can easily handle irregular data and has been developed over the past 30 years. 
For a general introduction to FDA, refer to introductory texts on applied and theoretical aspects such as \citet{citeulike:11611491}, \citet{horvath2012inference}, and \citet{hsing2015theoretical}, as well as review articles like \citet{wang2016functional} and \citet{koner2023second}.

In FDA, there are two common types of longitudinal data: one where each subject has a sufficient number of observations and another where the number of observations is very limited. 
Generally, the former is referred to as \textit{densely observed data} or just \textit{dense data}; the latter as \textit{sparsely observed data} or just \textit{sparse data} \citep{james2003clustering,zhang2016sparse}. 
Although there is no formal definition of densely/sparsely observed data across the fields, classification is sometimes based on the number of observation points (see \citealp{zhang2016sparse}). 
In this paper, we define a dense setting as a situation where each subject has a sufficient number of measurement points for pre-smoothing to work effectively and a sparse setting as one where this is not the case. 

For clustering densely observed data, pre-smoothing methods for each subject are effective.
For example, in \citet{abraham2003unsupervised}, 
we first obtain a functional representation for each subject through pre-smoothing with B-spline basis functions. 
We then apply $k$-means clustering to the coefficient vectors of these functional representations based on the basis functions.
Similarly, various methods combining the pre-smoothing with variants of the $k$-means algorithm have been developed. 
A robust method using cubic B-spline basis and trimmed $k$-means has been introduced, offering resistance to outliers \citep{garcia2005proposal}. 
To capture multi-resolution features, some methods incorporate wavelet transforms \citep{antoniadis2013clustering, delaigle2019clustering, lim2019multiscale}. 
\citet{chiou2007functional} propose the clustering method that uses predictors conditioned on clusters through the truncated Karhunen-Lo\`eve (KL) expansion expansion.
The idea of $k$-centers via subspace projection has been further developed to the clustering method based on a shape function model with random scaling effects \citep{chiou2008correlation}.

\begin{figure}[!tb]
	\begin{center}
		\renewcommand{\arraystretch}{1}
		\vspace{0.5cm}
		\begin{tabular}{c}
			\includegraphics[width=7cm]{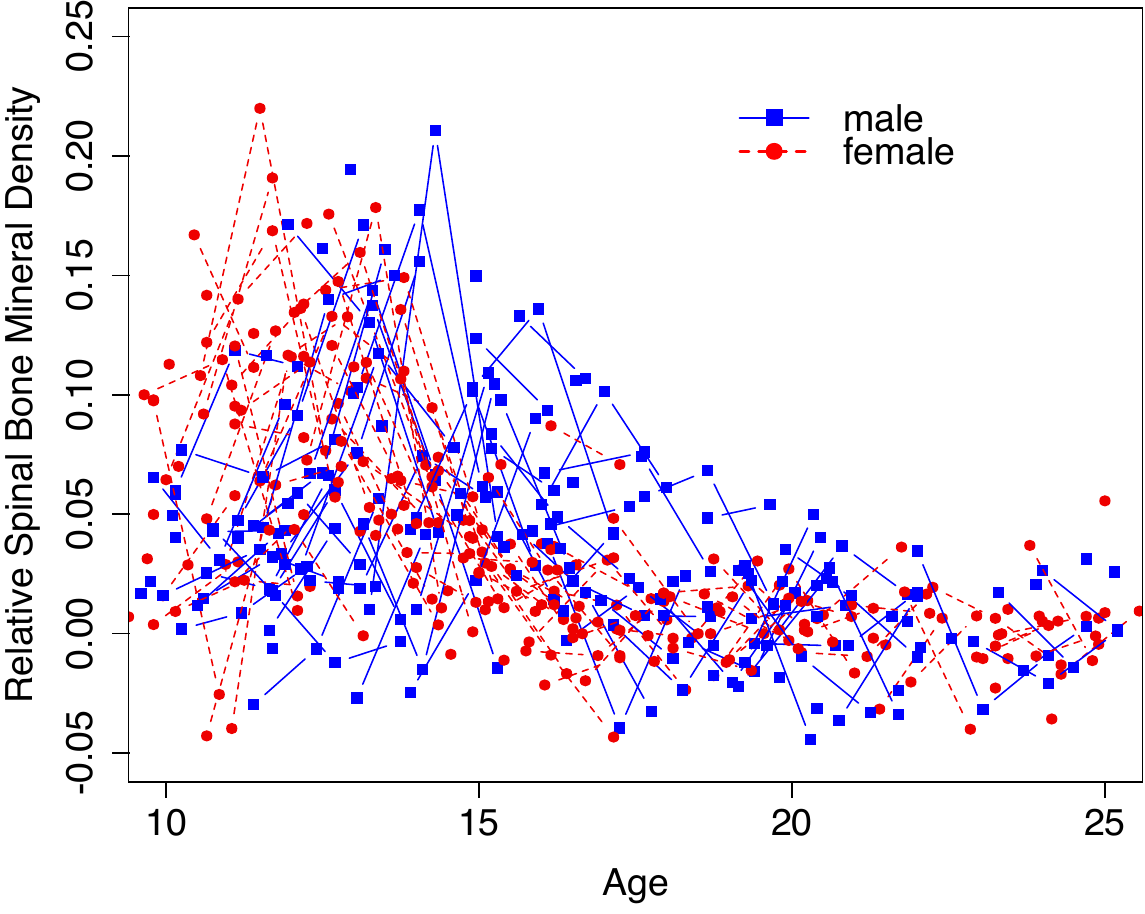}
		\end{tabular}
		\caption{An example of sparsely observed data from \citet{bachrach1999bone}.
  }
		\label{fig:example-sparse-data}
	\end{center}
\end{figure}

In sparse settings, however, applying FDA becomes challenging since the number of measurements per subject is very small.
For example, when applying a pre-smoothing method like \citet{abraham2003unsupervised}, the approximation accuracy of the curves can deteriorate significantly or even become infeasible, leading to poor clustering results based on such inaccurate approximations.
Figure \ref{fig:example-sparse-data} shows an example of such sparsely observed data.
The dataset, which is obtained from \citet{bachrach1999bone}, recodes the difference in spinal bone mineral density measured at two consecutive visits, divided by the average.
In the data, each subject has, on average, fewer than two measurements, and the measurement points vary among subjects.
Many classical FDA methods cannot be applied to such data.
To address this limitation, many methods for sparsely observed data have been developed in the last two decades.
There are many kinds of research on the problem of estimating the mean and covariance functions \citep{yao2005functional,zhou2008joint,zhang2016sparse,li2020fast}, functional regression analysis \citep{yao2005functionalregresstion,csenturk2008generalized,he2010functional}, functional principal component analysis \citep{yao2005functional,hall2006properties,yao2006penalized,paul2009consistency}, and functional discriminant analysis \citep{james2001functional,park2022sparse}.
For a helpful summary of methods for sparsely observed data, see \citet{wang2016functional}.


To address this issue, several clustering methods also have been developed for sparsely observed data.
As a pioneering work, 
\citet{james2003clustering} introduce a general functional clustering model,
assuming that each individual follows a cluster-specific Gaussian process. 
In their model, mean and covariance functions are represented by a finite set of basis functions.
While they used the cubic B-spline basis functions, 
other studies have employed different basis functions (e.g., \citealp{giacofci2013wavelet,ng2006mixture}).
Other model-based clustering methods have been developed, including approaches that assume finite dimensional subspaces for cluster structure \citep{bouveyron2011model,bouveyron2015discriminative} and methods using functional principal components \citep*{jacques2013funclust,rivera2019robust}.
In contrast to the model-based approaches, \citet{peng2008distance} propose a method for estimating $L_{2}$ distances between subjects based on the truncated KL expansion, followed by applying the ordinary clustering methods to these estimated distances.

This paper proposes a fast $k$-means clustering method for sparsely observed data, which can be viewed as a simplification of the general functional model of \cite{james2003clustering}.
The $k$-means clustering is often formulated as a problem of estimating the cluster centers \citep{pollard1981strong,linder2002learning,biau2008performance,levrard2015nonasymptotic}.
Similarly, in our proposed method, we formulate the optimization problem to estimate cluster centers in a function space and represent cluster centers by the basis expansion.
Unlike \citet{abraham2003unsupervised}, our method does not smooth individual observations, making it applicable even to sparsely observed data.
This approach allows the cluster center function to be estimated using all discretely observed values from subjects within the cluster, allowing effective estimation even when each subject has only a few observed values. 
Consequently, the proposed method can reliably identify groups of subjects from sparsely observed data based on the estimated cluster centers.


This paper makes three significant contributions.
First, we formulate the $k$-means clustering to be applicable to sparsely observed data, which can be viewed as a simplification of \cite{james2003clustering}.
Second, we prove the theoretical properties of the proposed method that clustering methods should have \citep{pollard1981strong,linder2002learning,biau2008performance,levrard2015nonasymptotic}.
Third, we evaluate the clustering performance of the proposed method for finite samples, especially for small samples, by numerical experiments and confirm that the clustering estimation accuracy is better than or equal to existing clustering methods for sparsely observed data.
In particular, due to the simplicity of the proposed method, it has been demonstrated that the execution time is significantly faster compared to the method of \citet{james2003clustering}, an existing method with good performance.
In addition to these contributions, we illustrate the effectiveness of the proposed method by applying it to spine bone density data \citep{bachrach1999bone}.
Furthermore, we have created an $\mathtt{R}$ package $\mathtt{fkms}$ that can implement the proposed method.
It should be noted that the proposed method is not only applicable to sparsely observed data but also to densely observed data, making it a versatile method suitable for various situations.

The remainder of this paper is organized as follows.
In Section \ref{sec:proposed-method}, the proposed method is described, and then, in Section \ref{sec:theoretical-results}, we provide the theoretical properties of the proposed estimator.
Section \ref{sec:exmeriments} explains numerical experiments using artificial data to investigate the performance of the proposed method on finite samples.
In Section \ref{sec:real-data}, we apply the proposed method to spinal bone density data to illustrate the usefulness of the proposed method.
Finally, in Section \ref{sec:concluding-remarks}, some discussions are given, and directions for further research are mentioned.

\section{Proposed method}
\label{sec:proposed-method}

Let $T$ be a real-valued random variable with a continuous probability distribution $P_{T}$, which has the positive density on a closed interval $\mac{T}$.
Let $T_{ij}$ $(i=1, \dots, n; j=1, \dots, N_{i})$ represents the time point of the $j$th measurement for the $i$th subject, where each $T_{ij}$ is an i.i.d.\ copy of $T$.
We will denote the set of all square-integrable functions with respect to the measure $P_{T}$  on $\mac{T}$ by $L_{2}(P_{T})$.
Throughout, for any $f,g\in L_2(P_T)$, we define the inner product by $\left<f,g\right>_{P_T}=\int_{\mac{T}}f(t)g(t)dP_T(t)$ and the associated norm by $\|f\|_{P_T}=\left<f,f\right>_{P_T}^{1/2}$.
Let $X$ be an $L_2(P_T)$-valued random element.
we consider the clustering of $L_{2}(P_{T})$-valued random observations $X_{1},\dots,X_{n}$, that are i.i.d.\ copies of $X$, 
from the sparse observations $\{T_{ij},X_{i}(T_{ij});i=1,\dots,n,j=1\dots,N_{i}\}$.
Note that we will assume the independence among $T$, $T_{ij}$, $X$, and $X_{i}$, while the discretely observed measurement $X_{i}(T_{ij})$ may depend on $X_{i}(T_{ij'})$ for $j\neq{}j'$.

The aim of clustering is to assign each subject to one of a finite number of $K$ groups.
The $k$-means clustering in the space $L_2(P_T)$ finds the optimal cluster centers to minimize the expected square loss function
\begin{align}
	\Psi(\bs{f})
	=\mbb{E}\left[
	\min_{k=1,\dots,K}\|X-f_{k}\|_{P_T}^{2}
	\right]
	\label{eq:expected-square-loss}
\end{align}
over all possible choices of cluster centers $\bs{f}=\{f_{1},\dots,f_{K}\}$ with $f_k\in L_{2}(P_{T})$, $k=1,\dots,K$.
Then, each subject is assigned to the cluster with the closest cluster center. 
In practice, we cannot evaluate the loss function $\Psi$ since we do not know the distribution of $(X,T)$.
Instead, the empirical version of the loss $\Psi$ that can be calculated based on the observations has to be used.
In this paper, we propose the functional $k$-means clustering (abbreviated as FKM), which finds the optimal cluster centers by minimizing the following empirical loss function
\begin{align}
	\tilde{\Psi}_{n}(\bs{f})
	=\frac{1}{n}\sum_{i=1}^{n}\min_{k=1,\dots,K}\sum_{j=1}^{N_{i}}\frac{1}{N_{i}}
	\left\{X_{i}(T_{ij})-f_{k}(T_{ij})\right\}^{2}
	\label{eq:empirical-loss}
\end{align}
over all possible choices of cluster centers. 
The factor 
$1/N_i$ works as a weight for subject $i$, and it can be replaced with any other value.
For example, \citet{huang2002varying} have discussed in the functional regression framework that different choices of the weight for subject $i$ in the square loss lead to different convergence rates of the estimators.
Also, \citet{zhang2016sparse} have discussed the two types of weights, equal weight per observation (OBS) and equal weight per subject (SUBJ), in the estimation of mean and covariance functions. 
They show that the convergence properties are different between the two weighting schemes.
The estimator based on our loss function $\tilde{\Psi}$ can be considered the SUBJ scheme, and the same discussion could also be made for our method.

Let $\mac{F}$ be the set of all cluster centers $\bs{f}$.
Then, the set of population global optimizers will be denoted by $Q(\mac{F})=\{\bs{f}\in\mac{F}\mid\inf_{\bs{f}'\in\mac{F}}\Psi(\bs{f}')=\Psi(\bs{f})\}$.
We write $Q=Q(\mac{F})$ for simplicity.
The goal of the analysis is to estimate $\bs{f}^{*}\in{}Q$ from observed data.
It is practically achieved by finding cluster centers $\bs{f}\in\mac{F}$ that minimize the empirical loss function $\tilde{\Psi}(\bs{f})$.
In this paper, we propose the estimation procedure based on the basis expansion of cluster centers.
That is, by using basis functions $\{\phi_{l}(t)\}_{l=1,\dots,m}$ on $\mac{T}$, each cluster center $f_{k}$ is represented as $f_{k}(t)=\sum_{l=1}^{m}\beta_{kl}\phi_{l}(t)=\bs{\beta}_{k}^{\top}\bs{\phi}(t)$ where $\bs{\beta}_{k}=(\beta_{k1},\dots,\beta_{km})^{\top}$ and $\bs{\phi}(t)=(\phi_{1}(t),\dots,\phi_{m}(t))^{\top}$.
Note that we do not assume that the population global optimizer $\bs{f}^{*}\in{}Q$ can be represented by using basis functions in this manner.

We will denote realizations of the sparse observations by $\{t_{ij},x_{i}(t_{ij});\ i=1,\dots,n,j=1,\dots,N_{i}\}$.
Also, we will denote the set of indices of subjects who belong to cluster $k$ by $\mac{C}_{k}$ in the sample.
Then, the proposed algorithm to find optimal cluster centers is summarized in Algorithm \ref{alg:proposed-algorithm}.
\begin{algorithm}[!tb]
	\caption{Clustering algorithm of FKM}
	\label{alg:proposed-algorithm}
	\begin{algorithmic}[1]
		\STATE Set initial cluster assignment to each subject.
		\STATE For each $k=1,\dots,K$, the coefficient is updated as
		\begin{align*}
			\bs{\beta}_{k}
			=\argmin_{\bs{\beta}\in\mbb{R}^{m}}\left\{
                    \sum_{i\in\mac{C}_{k}}\sum_{j=1}^{N_{i}}\frac{1}{N_{i}}
			    (x_{i}(t_{ij})-\bs{\beta}^{\top}\bs{\phi}(t_{ij}))^{2}
                    +\lambda_k\int\left(\frac{d^2 f(t)}{dt^2}\right)^2dt
                    \right\}
		\end{align*}
            where $f(t)=\bs{\beta}^\top\bs{\phi}(t)$, and $\lambda_k$ is a smoothing parameter.
		Then, calculate cluster centers by $f_{k}(t)=\bs{\beta}_{k}^{\top}\bs{\phi}(t)$.
		\STATE Assign each subject to the closest cluster that is obtained as
		\begin{align*}
			k=\argmin_{\ell=1,\dots,K}
			\sum_{j=1}^{N_{i}}\left\{x_{i}(t_{ij})-f_{\ell}(t_{ij})\right\}^{2}.
		\end{align*}
		\STATE Step 2 and Step 3 are repeated until convergence.
	\end{algorithmic}
\end{algorithm}
In the algorithm, we first give an initial cluster assignment for all subjects.
It is possible to give an initial assignment randomly, or it can be determined based on prior knowledge.
Then, for each cluster $k$, the coefficient $\bs{\beta}_{k}$ is obtained by minimizing the penalized least square criterion based on the observed data $x_{i}(t_{ij})$ within the cluster.
In nonparametric regression, regularized regression may be essential to avoid overfitting, especially when the model is highly flexible.
The penalty on the second derivative of the basis functions can help to control the complexity of the model and improve its generalized performance.
Therefore, we use the regularized regression formulation in the proposed algorithm.
Cluster centers are updated by using the obtained coefficients.
Next, assign each subject to the cluster with the closest center evaluated at the observed times.
The update to the cluster centers and the re-assignment of subjects will be repeated until convergence.
The convergence can be evaluated using various criteria.
In this paper, when the assignment of subjects does not change, we conclude that the algorithm has converged.

From Algorithm \ref{alg:proposed-algorithm}, the proposed algorithm seems similar to the ordinary $k$-means clustering.
Then, as is also the case with the $k$-means clustering, we need to set the initial values before we run the algorithm;
the obtained clustering results depend on initial values.
There would be various options to determine initial values.
A simple and popular way is to use several random initial starts; then, we choose the optimal estimator that minimizes the empirical loss function in (\ref{eq:empirical-loss}) among candidates.
In this paper, we adopt this method for initialization.

It should be noted that in Algorithm 1, if each subject's curve can be represented by a basis expansion, i.e., assuming $X_i(t)=\bs{\beta}_i^\top\bs{\phi}(t)$, $i=1,\dots,n$, and there are no penalty terms, then our proposed method is equivalent to the $k$-means method proposed by \citet{abraham2003unsupervised}.
However, in the case of sparsely observed data, such basis expansion cannot be performed on the data, and thus, their method cannot be applied.

The proposed method can be viewed as a simplification of the functional clustering model introduced by \citet{james2003clustering}. 
Specifically, under this model, our approach assumes the variance-covariance matrix of discretely observed data to be an identity matrix and employs a classification likelihood approach for estimation. 
In contrast, although \citet{james2003clustering} also operate within the functional clustering model, they use a different estimation approach.
Their method combines mixed-effects and normal mixture models, with parameters estimated via the EM algorithm. 
Their approach is more comprehensive and complex than our estimation method; this added complexity comes with the drawback of requiring substantial computational time (see Section \ref{sec:exmeriments}).

For the implementation of the proposed method, the type and number of basis functions, as well as the number of clusters, must be predetermined.
There are several choices for the type of basis functions, for example, Fourier, B-spline, and wavelet basis functions.
We can use principal components curves estimated through functional principal components analysis for sparse longitudinal data \citep{yao2005functional}.
Then, the parameters required for the basis functions, the number of clusters, and the values of smoothing parameters must be selected.
These tuning parameters can be jointly determined by the standard procedures based on the within-cluster variance for selecting the number of clusters, which include CH index \citep{calinski1974dendrite}, Silhouette index \citep{kaufman2009finding}, cross-validation (CV) for clustering \citep{wang2010consistent}, and bootstrapping \citep{fang2012selection}.
With the parameters for the basis functions and smoothing parameters being fixed, we may use the Gap statistic \citep{tibshirani2001estimating} and Jump statistic \citep{sugar2003finding} to determine the number of clusters.
Note that, as will be discussed in the next section, our proposed estimator of cluster centers has consistency when the smoothing penalty is not applied.
Thus, in this case, it is possible to achieve the consistent selection of the tuning parameters using the CV method proposed by \citet{wang2010consistent}.
In the real data analysis presented in Section \ref{sec:real-data}, 
with a sufficiently large number of basis functions, the values of the smoothing parameters are selected based on the CV method.

\section{Theoretical results}
\label{sec:theoretical-results}

In this section, we establish the consistencies of the loss function evaluated at the proposed estimator and the estimated cluster centers.
To simplify the discussion, in Step 2 of Algorithm 1, we focus on the case where the smoothing parameters are fixed at 0, and thus, the ordinary least squares method is applied.
\citet{pollard1981strong} provided the consistency of the cluster centers for the $k$-means clustering in the Euclidean space, and then various theoretical results have been developed \citep{linder1994rates,linder2002learning,fischer2010quantization,clemenccon2014statistical}.
The theoretical properties of $k$-means clustering in Hilbert spaces have also been developed \citep{biau2008performance,fischer2010quantization,levrard2015nonasymptotic}.
These studies deal with clustering for continuously observed functional data $X$.
In practice, however, functional data are discretely observed, and therefore, existing theoretical results cannot be applied directly to our setting.
We will demonstrate the consistency of the proposed estimator for discretely observed data.
The proofs of all theoretical results can be found in the Appendix.

Without loss of generality, we consider $\mac{T}=[0,1]$.
For simplicity, we assume that for some $M>0$, $\|X\|_{\infty}=\sup_{t\in\mac{T}}|X(t)|\le{}M$ a.s.
We assume that $\{\phi_{j}\}_{j=1,\dots,\infty}$ is a complete orthonormal system (CONS) in $L_2(P_T)$ and bounded: 
there exists some $C_\phi>0$ such that $\|\phi_{j}\|_{\infty}=\sup_{t\in\mac{T}}|\phi_{j}(t)|\le{}C_\phi$ for all $j\in \mathbb{N}$.
For the sake of simplicity, 
we assume that $P_T$ is the uniform distribution on $\mac{T}$.

We will denote the number of basis functions $m$ as $m_n$ to clarify the relationship with the sample size $n$.
Let $G_n$  be the linear space spanned by $\{\phi_{1},\dots,\phi_{m_n}\}$ and we write the direct products of $K$ linear spaces $G_n$ as $G_n^{K}=\prod_{k=1}^{K}G_n$.
For a constant $C>0$, let $G_n(C)=\{f\in{}G_n\mid\|f\|_{P_T}\le{}C\}$ and $G_n^{K}(C)=\prod_{k=1}^{K}G_n(C)$.
The set of the cluster centers $\{f_{1},\dots,f_{K}\}$ that are elements in 
$G_n(C)$ is expressed by
\begin{align*}
	\mac{F}_{G_n(C)}
	 & =\left\{\bs{f}\in\mac{F}\,\big|\, \bs{f}=(f_{1},\dots,f_{K})^\top\in{}G_n^K(C)\right\}.
\end{align*}
The set of all optimal cluster centers in $\mac{F}_{G_n(M)}$ that minimize the expected loss $\Psi$ in (\ref{eq:expected-square-loss}) is denoted by
\begin{align*}
	Q(\mac{F}_{G_n(M)})
	 & =\left\{\bs{f}\in\mac{F}_{G_n(M)}\,\Big|\,\inf_{\bs{f}'\in\mac{F}_{G_n(M)}}\Psi(\bs{f}')=\Psi(\bs{f})\right\}.
\end{align*}
Similarly, the set of cluster centers in $\mac{F}_{G_n(M)}$ that minimize the empirical loss $\tilde{\Psi}_{n}$ in (\ref{eq:empirical-loss}) is denoted by
\begin{align*}
	\widetilde{Q}_n(\mac{F}_{G_n(M)})
	 & =\left\{\bs{f}\in\mac{F}_{G_n(M)}\,\Big|\,\inf_{\bs{f}'\in\mac{F}_{G_n(M)}}\tilde{\Psi}_{n}(\bs{f}')=\tilde{\Psi}_{n}(\bs{f})\right\}.
\end{align*}

First, we can obtain the consistency in which the expected loss function evaluated at the proposed estimator converges to the optimal value of the loss function.
Write $D_n=n^{-1}\sum_{i=1}^n N_i^{-1/2}$.


\begin{Theo}
    \label{Thm:consistency-loss}
    Let $\bs{f}^{*}\in{}Q$ be any population global optimizer, and let $\tilde{\bs{f}}_{n}^{*}\in{}\widetilde{Q}_n(\mac{F}_{G_n(M)})$ be any empirical global optimizer.
    If $D_n^2 m_n^2 \log n\ra 0$, then as $n\ra\infty$, we have
	\begin{align*}
		\Psi(\tilde{\bs{f}}_{n}^{*})-\Psi(\bs{f}^{*})
		\ra 0\ \ \text{a.s.}
	\end{align*}
\end{Theo}

\vspace{10pt}
From this theorem, when the sample size is sufficiently large, the proposed estimators of cluster centers are expected to minimize the expected loss function.
Similar results have been reported for the ordinary $k$-means clustering in the Euclidean space \citep{linder1994rates,linder2002learning} and in the Hilbert space with continuously observed data \citep{biau2008performance}.
As far as we know, this is the first result of its kind for the clustering of discretely observed functional data.
The condition $D_n^2 m_n^2 \log n\ra 0$ represents the dependency of the number of basis functions $m_n$ and the number of measurement points per subject $N_i$ on the sample size $n$.

Next, we establish the consistency of the estimated cluster centers.
The distance between two finite subsets $A$ and $B$ of $L_{2}(P_{T})$ can be measured by the Hausdorff metric that is defined as
\begin{align*}
	d_{H}(A,B)
	=\max\left\{
	\max_{a\in{}A}\left\{\min_{b\in{}B}\|a-b\|_{P_T}\right\},
	\max_{b\in{}B}\left\{\min_{a\in{}A}\|a-b\|_{P_T}\right\}
	\right\}.
\end{align*}
Then, the distance between any cluster centers $\bs{f}\in\mac{F}$ and the set of cluster centers $A\subset\mac{F}$ is defined as
\begin{align*}
	d(\bs{f},A)
	=\inf\{d_{H}(\bs{f},\bs{g})\mid{}\bs{g}\in{}A\}.
\end{align*}

\begin{Theo}
    \label{Thm:consistency-estimator}
    Let $\tilde{\bs{f}}_{n}^{*}\in{}\widetilde{Q}_n(\mac{F}_{G_n(M)})$ be any empirical global optimizer.
    If $D_n^2 m_n^2 \log n\ra 0$, then as $n\ra\infty$, we have
    \begin{align*}
	d(\tilde{\bs{f}}_{n}^{*},Q)
	\ra0\ \ \text{a.s.}
    \end{align*}
\end{Theo}

As mentioned in Section 1, the proposed method can be applied not only to sparse observational data but also to dense data. 
For dense data, \citet{abraham2003unsupervised} proposed a $k$-means method based on the basis function expansion of each subject’s data $X_i$ and demonstrated the consistency of their estimator for cluster centers. 
Compared to their result, Theorem \ref{Thm:consistency-estimator} guides the appropriate size of the number of basis functions $m_n$ in relation to the sample size $n$ and the number of measurements $N_i$


\section{Experiments}
\label{sec:exmeriments}


The finite sample performance of the proposed method, FKM, was investigated through simulations.
First, we evaluate the relatively large sample behavior of the estimated cluster centers using FKM.
Then, to evaluate the finite sample property on cluster recovery, FKM is compared with commonly used existing clustering methods for sparsely observed data: functional clustering model \citep{james2003clustering} and clustering based on the estimated distance proposed by \citet{peng2008distance}.
Following the terminology used in \citet{yassouridis2018generalization}, in the remainder of this paper, we will refer to those methods as \textit{FCM} and \textit{distclust}, respectively.
Moreover, the proposed method is a $k$-means-like simple algorithm, and it is expected to analyze relatively large datasets within a reasonable amount of time.
Therefore, we evaluate the execution time of the proposed method.
Unless otherwise indicated, all analyses were performed using \texttt{R} version 4.3.3 \citep{Rmanual}.


We used the following data-generating process.
This simulation deals with the two-cluster setting, i.e., $K=2$.
The number of measurements $N_{i}$ of subject $i$ was randomly drawn from the binomial distribution with expected value $N_{tp}$. 
Then, the time points $t_{ij}$ ($j=1,\dots,N_{i}$) were randomly drawn from the uniform distribution on $[0,1]$.
If the generated number of measurements was less than 2, it was set to 2.
\mymemo{
	【提案手法は時点数が1でも利用可能であるが、他の方法（特に、distclust）を適用可能とするために、最小の時点数を2に設定した。】
}
Next, measurements on $t_{ij}$ of subject $i$ who belonged to cluster $k$ were generated through the following equation
\begin{align}
	X_{i}(t_{ij})
	=\sum_{u=1}^{40}(u^{-1}(Z_{i}-1)+\mu_{ku})\sqrt{2}\sin(\pi{}ut_{ij})+\epsilon_{ij}
	\label{eq:simulation_data}
\end{align}
where $Z_{i}$ was randomly drawn from the exponential distribution with rate $\lambda=1$ and $\epsilon_{ij}$ was randomly drawn from the normal distribution $N(0,\sigma^{2})$.
The equation of data generation is similar to that of Delaigle and Hall (2012).
The difference between clusters was expressed in $\mu_{ku}$, $u=1,\dots,40$, which were determined as
\begin{align*}
	\bs{\mu}_{1}
	 & =(0.5, -0.2, 1, -0.5, 0, -0.7, \bs{0}_{34}^{\top})^{\top},     \\
	\bs{\mu}_{2}
	 & =(0, -0.75, 0.75, -0.15, 1.4, 0.1, \bs{0}_{34}^{\top})^{\top}.
\end{align*}

Three experimental factors are considered: sample size $n$, the expected number of measurements $N_{tp}$, and the size of the error variance $\sigma^{2}$.
The sample size was set at $n=50, 100, 200$, and $400$.
Note that the balanced case between clusters was considered; that is, the sample size belonging to each cluster was $n/2$.
We considered three levels of expected measurements $N_{tp}=3$, $5$, and $10$, in which the smaller the value, the more sparse the condition is considered to be.
Finally, the size of error variances was set at three levels: $\sigma=0.1$, $1.0$, and $2.0$.
The experimental design was fully crossed, with 100 replications per cell, yielding $4\times3\times3\times100=3600$ simulated data sets.
Examples of the generated observations are shown in Figure \ref{fig:simulation_data_example}.
Population optimal cluster centers in the figure were calculated by the $k$-means clustering for the dense equispaced measurements that were generated through (\ref{eq:simulation_data}) without the error term $\epsilon_{ij}$ when $n=\text{10,000}$ and the number of measurements was 400 for all subjects.

\begin{figure}[!tb]
	\begin{center}
		\renewcommand{\arraystretch}{1}
		\vspace{0.5cm}
		\begin{tabular}{c}
			\includegraphics[width=17cm]{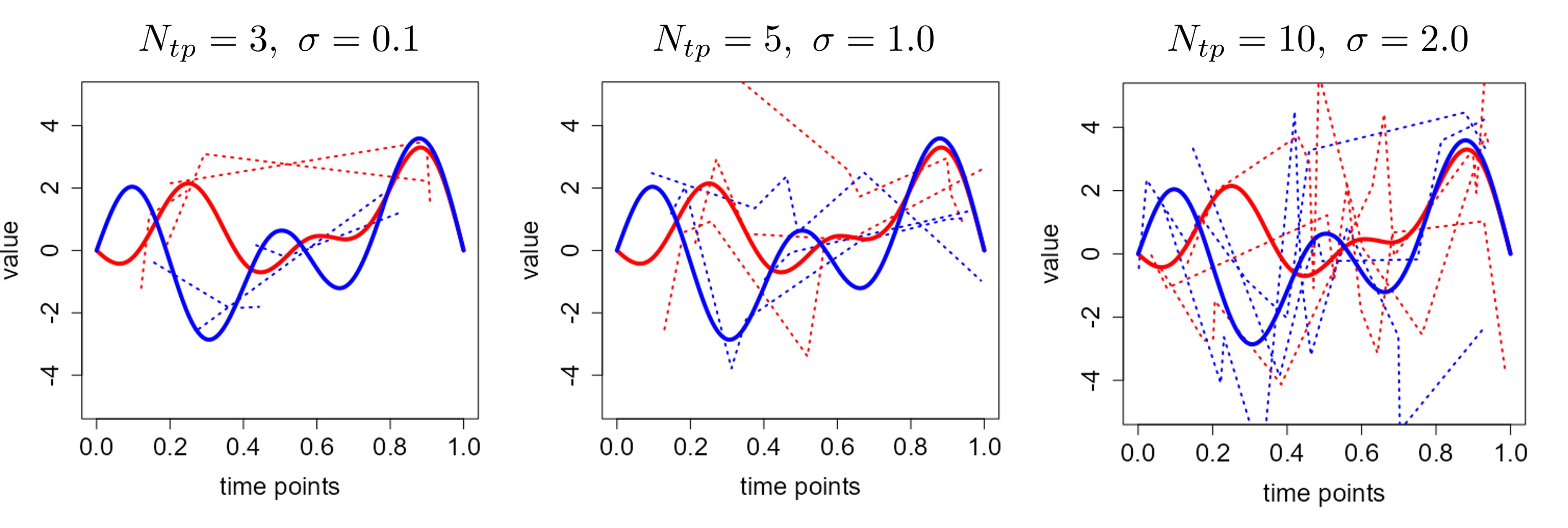}
		\end{tabular}
		\caption{
            Generated observations for $(N_{tp}, \sigma) = (3, 0.1), (5, 1.0), (10, 2.0)$ (from left to right); 
            each panel representing optimal cluster centers (thick lines) and three subject curves per cluster connecting observed points (thin dotted lines);
            clusters in red and blue.
  }
  		\label{fig:simulation_data_example}
	\end{center}
\end{figure}

We used the \texttt{R} package \texttt{fkms} to run FKM.
In these experiments, we used Fourier and B-spline basis functions as those required for FKM.
Those are referred to as FKM-f and FKM-b, respectively.
The number of basis functions was fixed at fifteen so that it was sufficient to represent the mean functions well enough.
As in Section \ref{sec:theoretical-results}, the smoothing parameter in Step 2 of Algorithm 1 was set to 0 for this simulation to simplify the discussion.
As described above, the results of FKM depend on the initial setting of the cluster assignment.
Thus, we used a hundred random initial values for the first cluster assignment in Step 1 of Algorithm \ref{alg:proposed-algorithm}.
The result with the smallest empirical loss \eqref{eq:empirical-loss} was adopted as the final result.

\begin{figure}[!tb]
	\begin{center}
		\renewcommand{\arraystretch}{1}
		\vspace{0.5cm}
		\begin{tabular}{c}
			\includegraphics[width=18cm]{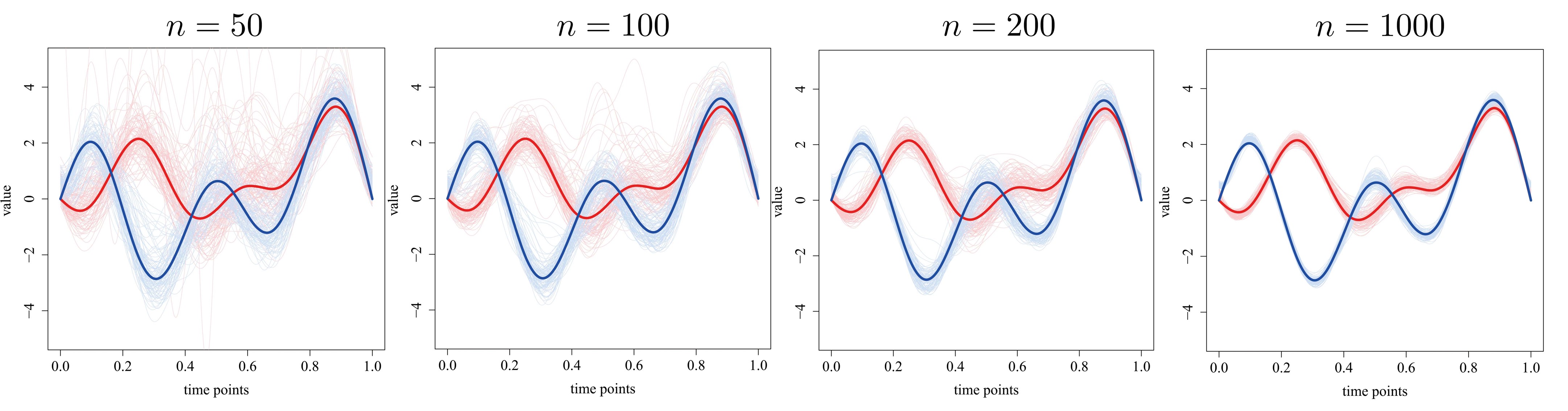}
		\end{tabular}
   \caption{
   Estimated cluster centers (light lines) across 100 replications and 
   optimal cluster centers (bold lines) at the expected time points $N_{tp}=10$ and error variance $\sigma=0.1$; 
   panels representing sample sizes $n=50$, 100, 200, and 1000 (from left to right); clusters in red and blue.
   }
		\label{fig:simulation_convergence}
	\end{center}
\end{figure}

\vspace{10pt}
First, we show the results of the relatively large sample behavior of the proposed estimator.
Figure \ref{fig:simulation_convergence} shows the estimated cluster centers by FKM with Fourier basis functions for 100 repetitions at each combination of conditions.
This experiment was conducted with a maximum sample size of 1000.
Four cases of $n$ are shown when the expected number of time points and the error variance are set at $N_{tp}=10$ and $\sigma=0.1$.
As the sample size increased, 
the accuracy of the estimation of cluster centers improved as expected.
In other words, as stated in Theorem \ref{Thm:consistency-estimator}, the estimator of the proposed method converges to the truly optimal cluster center when the number of bases and the sample size are sufficiently large.
This trend was also true when the B-spline basis functions were used; see \textit{Supplementary Information} \ref{SI-subsec:cluster-centers-Bspline} for the results of the B-spline basis.
For reference, results under all the combinations of conditions are shown in \textit{Supplementary Information} \ref{SI-subsec:cluster-centers-Fourier}.

Next, we move to the finite sample property of cluster recovery compared with existing methods, FCM and distclust.
When implementing FCM, we used the clustering likelihood \citep{james2003clustering} and B-spline basis functions.
Here, we used the same number of basis functions as FKM.
The results of FCM depend on the first cluster assignment in the EM algorithm; hence, we used a hundred random initial starts to obtain the result of FCM, and the final estimate was selected as one that gave the largest likelihood among the hundred estimates.

For conducting distclust, the number of principal components used to estimate distances among subjects has to be determined.
We used the same number of principal components as that of the basis functions of FKM and FCM.
The clustering method used in distclust depends on the initial values for conducting the algorithm; hence, we used a hundred random initial starts for distclust, in which the final estimate was selected as one that gave the smallest within-cluster variance among the hundred estimates.
We conducted distclust using \texttt{R} package \texttt{funcy} with \texttt{R} version 3.6.3, as the package is incompatible with later versions.



\vspace{5pt}
Clustering performance was evaluated by correct classification rate (CCR) and adjusted Rand index (ARI; \citealp{hubert1985comparing}).
In clustering, cluster labels are indeterminate; hence, the error rate is defined by a minimum error mapping between the predicted classification and the known true clusters, and the CCR is defined as one minus the error rate.
The ARI has a maximum value of 1 in the case of a perfect cluster recovery and a value of 0 when the true cluster membership and the estimated cluster membership do not coincide more than would be expected by chance.

\begin{table}[!tb]
	\caption{Correct classification rates in percentages averaged over 100 replications under the combinations of conditions $(N_{tp},\sigma,n)$; $N_{tp}$, $\sigma$, and $n$ denote the expected number of time points, the standard deviation of the error term, and sample size, respectively.}
	\label{tab:simulation-results}
	\centering
	\begin{tabular}{cccccccccccccc}
		\toprule
		                            &           &                         & \multicolumn{3}{c}{$\sigma$=0.1} & \multicolumn{4}{c}{$\sigma$=1.0} & \multicolumn{4}{c}{$\sigma$=2.0}                                                         \\
		\cmidrule{3-14}    $N_{tp}$ & Methods   & \multicolumn{4}{c}{$n$} & \multicolumn{4}{c}{$n$}          & \multicolumn{4}{c}{$n$}                                                                                                     \\
		                            &           & 50                      & 100                              & 200                              & 400                              & 50   & 100  & 200  & 400  & 50   & 100  & 200  & 400  \\
		\midrule
		3                           & FKM-f     & 66.3                    & 68.3                             & 72.0                             & 73.3                             & 63.8 & 65.4 & 67.6 & 69.3 & 59.4 & 58.7 & 59.9 & 62.6 \\
		                            & FKM-b     & 67.1                    & 68.3                             & 72.2                             & 73.6                             & 64.5 & 65.8 & 67.8 & 69.4 & 59.4 & 60.2 & 60.2 & 62.1 \\
		                            & FCM       & 66.8                    & 68.7                             & 72.7                             & 75.9                             & 64.1 & 64.8 & 67.7 & 69.5 & 59.3 & 59.8 & 60.6 & 61.9 \\
		                            & distclust & 59.3                    & 60.5                             & 61.9                             & 63.6                             & 58.1 & 58.3 & 58.5 & 60.9 & 57.2 & 56.1 & 55.6 & 56.6 \\
		\midrule
		5                           & FKM-f     & 71.1                    & 79.1                             & 81.0                             & 82.4                             & 66.5 & 72.2 & 75.7 & 77.7 & 61.6 & 64.3 & 66.7 & 68.8 \\
		                            & FKM-b     & 71.2                    & 79.1                             & 81.3                             & 82.4                             & 67.4 & 73.3 & 75.2 & 77.7 & 61.8 & 64.7 & 66.8 & 68.5 \\
		                            & FCM       & 74.9                    & 82.2                             & 84.4                             & 85.9                             & 65.8 & 74.8 & 78.7 & 81.0 & 62.1 & 65.0 & 66.8 & 70.0 \\
		                            & distclust & 64.0                    & 66.8                             & 70.3                             & 70.4                             & 62.0 & 64.3 & 67.1 & 67.0 & 58.4 & 58.1 & 59.8 & 60.8 \\
		\midrule
		10                          & FKM-f     & 82.9                    & 86.4                             & 89.0                             & 90.0                             & 79.7 & 82.6 & 85.2 & 86.1 & 72.2 & 74.7 & 77.1 & 78.3 \\
		                            & FKM-b     & 83.8                    & 86.2                             & 89.0                             & 90.0                             & 79.9 & 82.6 & 85.7 & 86.4 & 72.3 & 74.5 & 77.4 & 77.9 \\
		                            & FCM       & 90.9                    & 93.7                             & 94.4                             & 95.0                             & 84.7 & 87.8 & 90.4 & 91.1 & 71.9 & 76.3 & 80.1 & 82.8 \\
		                            & distclust & 77.8                    & 80.2                             & 80.8                             & 81.7                             & 73.4 & 75.3 & 78.9 & 78.4 & 66.6 & 67.7 & 68.8 & 71.3 \\
		\bottomrule
	\end{tabular}%
\end{table}%

Table \ref{tab:simulation-results} shows CCRs averaged over 100 replications under the combinations of conditions.
First, cluster recovery is better for all methods as the sample size or the expected number of measurements increases, and it is worse as the error variance increases.
When the number of measurements is small, the proposed method and FCM are equally good, and as the number of measurements increases, FCM gives slightly better results.
This result might be reasonable given the fact that FCM is a parametric method with a larger number of parameters to be estimated.
The proposed methods and FCM performed better than distclust in recovering clusters under all conditions.
Additionally, there was no significant difference between FCM-f and FCM-b for all conditions.
This indicates that the choice of basis functions did not affect the clustering results for the datasets used in this experiment.
For ARI, the same trend was observed  (see \ref{SI-subsec:ARI} in the \textit{Supplementary Information}).

Next, we compared the computational time required to perform FKM and FCM.
We set the sample sizes to $n=50$, 100, 200, 400, 600, 800, and 1000; we created ten samples for each combination of conditions.
Using only one initial value, we analyzed these ten samples with FKM using Fourier basis functions and FCM.
The settings are the same as those in the experiments above.
The numerical experiments were conducted on a standard desktop PC equipped with an Intel(R) Core(TM) i9-10980XE CPU.
Furthermore, to improve computation time, we utilized the \texttt{RCpp} package in the implementation of both the proposed method and FCM.

Figure \ref{fig:implementation-median-time} shows the results when the error variance is $\sigma = 1.0$ (see \ref{SI-subsec:implementation-median-time} in \textit{Supplementary Information} for results of all $\sigma$ cases).
The median time (in seconds) taken for a single analysis using only one initial value is shown for each sample size, with error bars representing the maximum and minimum values across ten samples.
From left to right, the results are for the expected number of measurements $N_{tp} = 3$, 5, and 10.
The black dashed line represents the results of FCM, while the red solid line represents the results of FKM-f.
As the sample size increases, the computational time for FCM increases linearly.
In contrast, the computational time increase for FKM-f is minimal compared to FCM.
Additionally, there is a tendency for the computational time to increase as the expected number of measurements increases.

In general, to obtain stable results using FCM and FKM, a large number of initial values are required.
For example, \citet{steinley2006k} mentions performing thousands of random restarts in the context of the standard $k$-means method.
In the experiments shown in Table \ref{tab:simulation-results}, a hundred random initial values were necessary for each method to obtain reasonably good results; the required number of initial values may increase with the sample size and the number of measurements.
Furthermore, practical analyses require the selection of many tuning parameters, such as the number of clusters, the type of basis functions, and the number of bases, necessitating a significant amount of trial and error.
In scenarios where datasets include large sample sizes or many measurement points, the use of FCM becomes impractical.
However, the proposed method can be executed without issues with computational time under those conditions.

\begin{figure}[!tb]
	\begin{center}
		\renewcommand{\arraystretch}{1}
		\vspace{0.5cm}
		\begin{tabular}{c}
			\includegraphics[width=16cm]{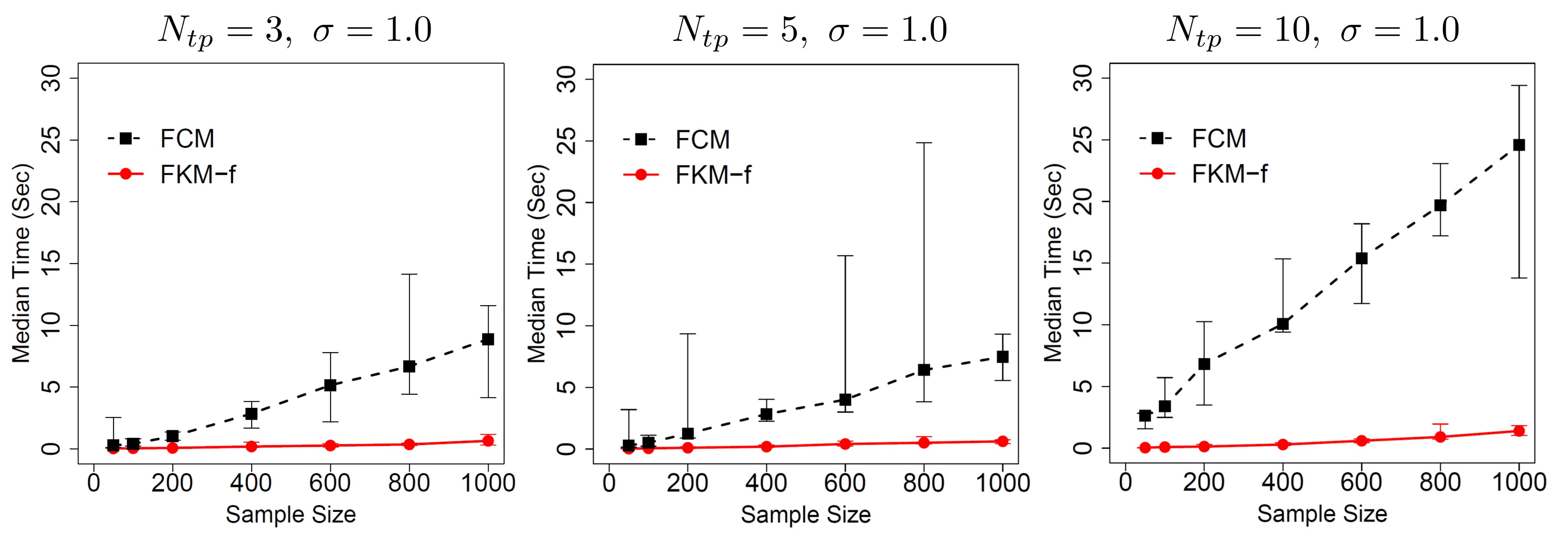}
		\end{tabular}
		\caption{
            Median analysis time (in seconds) with error bars showing the range (min-max) for each sample size at error variance $\sigma=1.0$ (black dashed line: FCM, red solid line: FKM-f);
            panels representing results for expected time points $N_{tp}=3$, 5, and 10 (from left to right);
            }
		\label{fig:implementation-median-time}
	\end{center}
\end{figure}

\section{Real data example}
\label{sec:real-data}


For illustration purposes, we applied the proposed method to the spinal bone mineral density data shown in Section \ref{sec:Introduction}.
The dataset is from \citet{bachrach1999bone}, consisting of relative spinal bone mineral density (BMD) measurements on 261 North American adolescents, with 116 males and 145 females.
Each value represents the difference in spinal BMD taken on two consecutive visits, divided by the average, and the age is the average age over the two visits.
As seen in Figure \ref{fig:example-sparse-data}, this dataset is quite sparse regarding longitudinal measurements, with the average number of measurement points being less than two.

Since bone density growth differs between males and females, gender can be considered a clustering factor.
The analysis aims to recover the clusters based on gender and estimate the cluster center functions. 

In this analysis, we used B-spline basis functions with the number of clusters set to two.
The number of basis functions was fixed at ten, and the smoothing parameters were determined by using the cross-validation (CV) with stability \citep{wang2010consistent} as mentioned in Section \ref{sec:theoretical-results}.
Note that the smoothing parameter was set to be the same across clusters and set to 75 based on the CV using candidate values of 0.01, 25, 50, 75, 100, 125, and 150.
We used a hundred random initial starts and selected the solution that minimized the empirical loss function as the final estimate.

\begin{figure}[!tb]
	\begin{center}
		\renewcommand{\arraystretch}{1}
		\vspace{0.5cm}
		\begin{tabular}{c}
			\includegraphics[width=7cm]{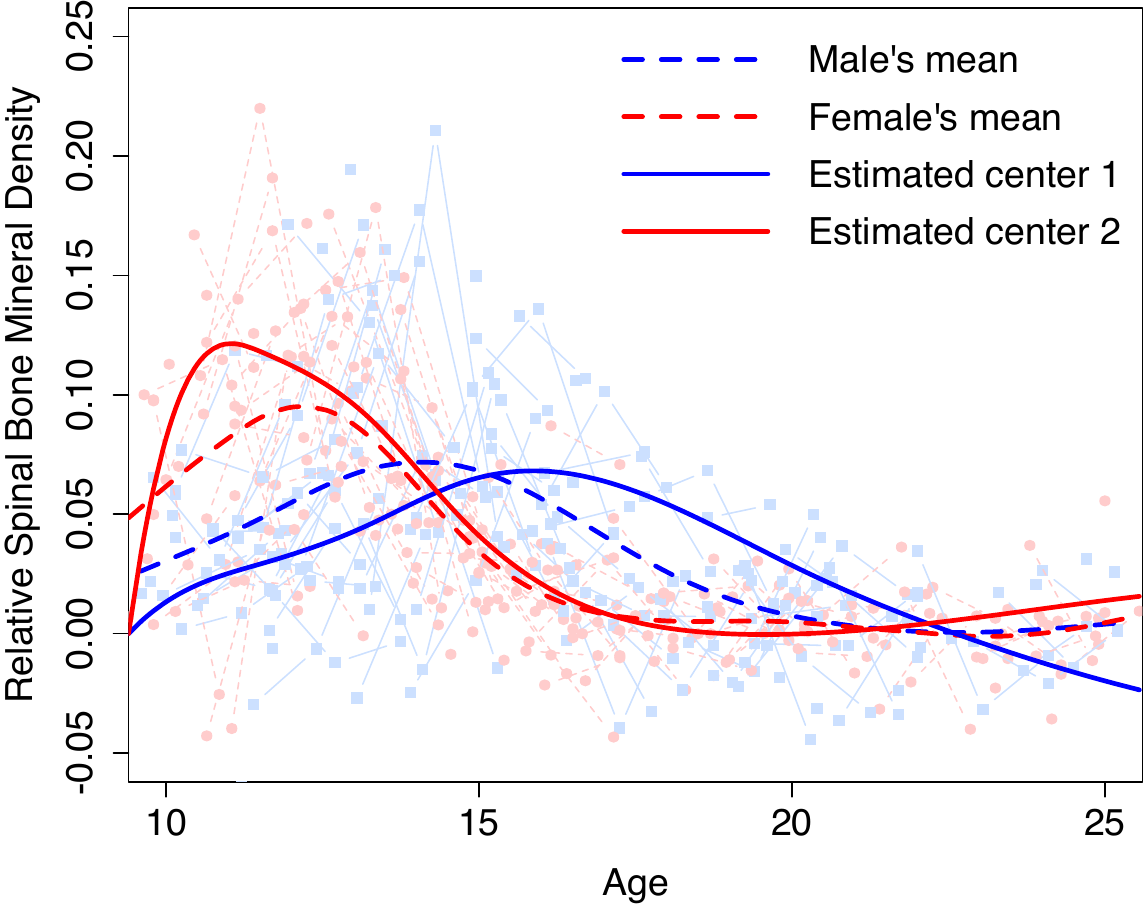}
		\end{tabular}
		\caption{
            Results of the proposed method on spinal bone mineral density data;
            light-colored lines connect values for each subject; 
            dashed lines: smoothing splines for each gender group (blue: male, red: female); 
            solid lines: the two estimated cluster centers.}
		\label{fig:bmd-estimates}
	\end{center}
\end{figure}

Figure \ref{fig:bmd-estimates} shows the estimated cluster centers as solid lines.
The figure also shows the cluster centers obtained by smoothing splines (dashed lines) for each gender using true labels, i.e., gender information.
The two sets of cluster centers are pretty similar, indicating that the proposed method estimates the relative bone density curves reasonably well, even for such sparse data.
From the estimates obtained by the proposed method, it can be inferred that the increase in bone density for females occurs around the age of 12, whereas for males, the growth in bone density occurs slightly later. 
This conclusion aligns with findings reported in studies such as \citet{bachrach1999bone}, which noted that the tempo of gains in BMD varied by gender.
It can be observed that, on the right tail, the estimated values by the proposed method are slightly higher or lower than the mean curves for each gender based on the true label. 
This issue is likely due to the smaller number of data points available in the right tail.

The correct classification rate (CCR) for the proposed method was $65.9\%$; when using the same number of basis functions, the CCRs for FCM and distclust were $59.0\%$ and $64.0\%$, respectively.
Therefore, the proposed method yielded slightly higher CCR values for this dataset.

\section{Conclusion}
\label{sec:concluding-remarks}

We have proposed a simple clustering method for sparsely observed longitudinal data.
This method can be seen as a direct extension of the classical $k$-means clustering to the situation with sparse observations.
The advantages of our proposed method include not only its consistent estimation of cluster centers and good clustering performance but also its low computational cost, making it suitable for large-scale data due to the simplicity of the objective function and algorithm.

Although this study focuses on longitudinal data defined on a subset of $\mathbb{R}$, it can be similarly applied to spatial data defined on subsets of $\mathbb{R}^2$
and even higher-dimensional domains, such as video data.
Additionally, while we have only considered univariate functional data, extending the method to deal with multivariate functional data is straightforward.

One limitation of the proposed method, similar to the classical $k$-means and many functional clustering methods, is its dependency on initial values (see Section \ref{sec:proposed-method}).
The simplest way to address this issue is to use a large number of initial values and select the best result.
However, the required number of initial values may depend on the data, suggesting a need for more robust methods against initial value sensitivity.
For classical $k$-means, the $k$-means++ algorithm \citep{arthur2007k} addresses this issue, but this algorithm is not directly applicable to sparsely observed data where each subject cannot be considered as a cluster centroid.

Furthermore, like the classical $k$-means clustering, the proposed method may be sensitive to outliers.
Robust clustering methods such as trimmed $k$-means \citep{cuesta1997trimmed} and robust $k$-means \citep{gallegos2005robust} have been proposed for classical $k$-means clustering, and similar approaches could be used to robustify the proposed method.

Although we have demonstrated the consistency of the cluster centers as a theoretical result, we have not yet investigated their asymptotic distribution. 
For the classical $k$-means, \citet{pollard1982central} proved the asymptotic normality of cluster centers.
Providing uniform confidence bands for the cluster centers would be helpful in practical clustering applications.
It seems feasible to construct confidence bands using bootstrap methods, but further examination of their asymptotic properties is necessary.
This aspect remains an important subject for future research.

\appendix

\renewcommand{\theequation}{A.\arabic{equation}}
\setcounter{equation}{0} 

\section{Appendix}

\subsection{Proof of Theorem \ref{Thm:consistency-loss}}

The empirical loss function when we would observe functional data completely, is written as
\ba
\Psi_{n}(\bs{f})
=\frac{1}{n}\sum_{i=1}^n\min_{k=1,\dots,K}\|X_i-f_k\|_{P_T}^2.
\ea
We give an upper bound for the difference in expected square losses between the estimated and truly optimal cluster centers.
Note that, by the definition of $\tilde{\bs{f}_n^*}$, $\tilde{\Psi}_n(\tilde{\bs{f}}_n^*)\le\tilde{\Psi}_n(\bs{f}_{G_n}^*)$ holds.
Therefore, for some $\bs{f}^{*}\in{}Q$, $\bs{f}_{G_n}^{*}\in{}Q(\mac{F}_{G_n(M)})$, 
and $\tilde{\bs{f}}_{n}^{*}\in{}\widetilde{Q}_n(\mac{F}_{G_n(M)})$, we have
\begin{align}
	 & \Psi(\tilde{\bs{f}}_{n}^{*})-\Psi(\bs{f}^{*})\notag       \\
	& \le\left[
		\Psi_{n}(\tilde{\bs{f}}_{n}^{*})-\tilde{\Psi}_{n}(\tilde{\bs{f}}_{n}^{*})+\tilde{\Psi}_{n}(\bs{f}_{G_n}^{*})-\Psi_{n}(\bs{f}_{G_n}^{*})
	\right]
	+	\left[
		\Psi(\tilde{\bs{f}}_{n}^{*})-\Psi_{n}(\tilde{\bs{f}}_{n}^{*})+\Psi_{n}(\bs{f}_{G_n}^{*})-\Psi(\bs{f}_{G_n}^{*})
		\right]
	+\Psi(\bs{f}_{G_n}^{*})-\Psi(\bs{f}^{*})
	\notag                                                       \\
	 & \le
	2\sup_{\bs{f}\in\mac{F}_{G_n(M)}}|\tilde{\Psi}_{n}(\bs{f})-\Psi_{n}(\bs{f})|
	+2\sup_{\bs{f}\in\mac{F}_{G_n(M)}}|\Psi_{n}(\bs{f})-\Psi(\bs{f})|
	+\{\Psi(\bs{f}_{G_n}^{*})-\Psi(\bs{f}^{*})\}.
	\label{eq:thm1-proof-bound}
\end{align}

From this expansion, we can see that if the three terms on the last line converge almost surely to 0, respectively, then the difference between the two expected square losses also converges almost surely to 0.
From Lemmas \ref{Lemm:diff_A1}, \ref{Lemm:diff_A2} and \ref{Lemm:diff_A3} in \textit{Supplementary Information} \ref{SI-sec:lemmas}, if the conditions described in the theorem hold, the three terms in \eqref{eq:thm1-proof-bound} converge almost surely to 0, which completes the proof.

\subsection{Proof of Theorem \ref{Thm:consistency-estimator}}

First, we provide the lemmas needed to prove the theorem.
The proofs of all lemmas are provided in \textit{Supplementary Information} \ref{SI-sec:proofs}.

For a constant $C>0$, we define
\begin{align*}
	\mac{F}_{(C)}
	=\{\bs{f}=(f_1,\dots,f_K)\in\mac{F}\mid\|f_k\|_{P_T}\le C,k=1,\dots,K\}.
\end{align*}
Then, it can be shown that $\Psi(\cdot)$ is uniformly continuous on $\mac{F}_{(M)}$.
\begin{Lemm}
	\label{Lemm:continuity-of-loss}
The expected loss $\Psi(\cdot)$ is uniformly continuous on $\mac{F}_{(M)}$.
\end{Lemm}
This result is the direct extension of the ordinary $k$-means clustering case, which is shown in \citet{terada2015strong}, to the functional setting.
Next, we state the continuity lemma for $\tilde{\Psi}_n$.
\begin{Lemm}
	\label{Lemm:continuity-of-empirical-loss}
	For any fixed $n$, the empirical loss $\tilde{\Psi}_n(\cdot)$ is uniformly continuous on $\mac{F}_{G_n(M)}$ almost surely.
\end{Lemm}
This fact can be proved	similarly to the case of $\Psi(\cdot)$, but we need to reformulate $\tilde{\Psi}_n(\cdot)$ using a semi-norm.
See \ref{subsec:proof-lemma-continuity-empirical-loss} in \textit{Supplementary Information}.

Furthermore, we describe the results on the relationship between the empirical square loss values evaluated with the optimal estimator on $\mac{F}_{G_n(M)}$ and the population expected square loss values evaluated with the truly optimal estimator.
\begin{Lemm}
	\label{lemm:lim-tilde-Psi-n}
        Under the assumptions in Theorem \ref{Thm:consistency-estimator},
	for any $\bs{f}_{G_n}^*\in Q(\mac{F}_{G_n(M)})$ and $\bs{f}^*\in Q$, as $n\ra\infty$,
	\begin{align*}
		\tilde{\Psi}_n(\bs{f}_{G_n}^*)\ra\Psi(\bs{f}^*)\ \
		\text{a.s.}
	\end{align*}
\end{Lemm}

Finally, we provide a result on the optimality of the truly optimal estimator on $\mac{F}_{G_n}$.
We define
\begin{align*}
    Q(\mac{F}_{G_n}\backslash B^o(Q,\epsilon))
	 & :=\{\bs{f}\in\mac{F}_{G_n}\mid\inf_{\bs{f}'\in\mac{F}_{G_n}\backslash B^o(Q,\epsilon)}\Psi(\bs{f}')=\Psi(\bs{f})\}
\end{align*}
where for $\bs{a}\in\mac{F}$, $B^o(\bs{a},\epsilon)$ denotes the open ball with center $\bs{a}$ and radius $\epsilon$, and for a subset $A\subset\mac{F}$ , $\cup_{\bs{a}\in A}B^o(\bs{a},\epsilon)$ is denoted by $B^o(A,\epsilon)$.
Note that, since $\Psi(\bs{f}_{G_n,\epsilon}^*)$ is monotonically decreasing with respect to the number of bases $m_n$ (and thus with respect to $n$) and bounded below, $\lim_{n\ra\infty}\Psi(\bs{f}_{G_n,\epsilon}^*)$ exists.

\begin{Lemm}
	\label{Lemm:inequality-eps}
        For any $\bs{f}_{G_n,\epsilon}^*\in 
        Q(\mac{F}_{G_n}\backslash B^o(Q,\epsilon))$ and $\bs{f}^*\in Q$,
	\begin{align*}
		\lim_{n\ra\infty}\Psi(\bs{f}_{G_n,\epsilon}^*)>\Psi(\bs{f}^*).
	\end{align*}
\end{Lemm}

\vspace{10pt}

Let's proceed to the proof of  Theorem \ref{Thm:consistency-estimator}.
First, for any $n$, consider some $\bs{f}_{G_n}^*\in Q(\mac{F}_{G_n(M)})$.
Then, we have
\begin{align*}
	\tilde{\Psi}_{n}(\tilde{\bs{f}}_{n}^{*})
	\le\tilde{\Psi}_n(\bs{f}_{G_n}^*)\ \
	\text{a.s.}
\end{align*}
Therefore,
\begin{align*}
	\limsup_{n\ra\infty}\tilde{\Psi}_{n}(\tilde{\bs{f}}_{n}^{*})
	\le\limsup_{n\ra\infty}\tilde{\Psi}_n(\bs{f}_{G_n}^*)\ \
	\text{a.s.}
\end{align*}
is established.
Furthermore, as inferred from Lemma \ref{lemm:lim-tilde-Psi-n}, we have
\begin{align*}
	\limsup_{n\ra\infty}\tilde{\Psi}_n(\bs{f}_{G_n}^*)
	=\Psi(\bs{f}^*)\ \
	\text{a.s.}
\end{align*}
Consequently,
\begin{align}
	\limsup_{n\ra\infty}\tilde{\Psi}_n(\tilde{\bs{f}}_n^*)
	\le\Psi(\bs{f}^*)\ \
	\text{a.s.}
	\label{eq:limsup-tilde-Psi-n_inequality}
\end{align}

Here, for any $\epsilon > 0$, let $\mac{F}_{G_n(M),\epsilon} := \{\bs{f} \in \mac{F}_{G_n(M)} \mid d(\bs{f},Q) \ge \epsilon\}$.
Now, according to Lemma \ref{Lemm:continuity-of-loss}, $\Psi$ is a continuous function on $\mac{F}_{(M)}$, hence it attains a minimum value on the compact subset $\mac{F}_{G_n(M),\epsilon}$.
Let us denote the point where this minimum is attained as 
$\bs{f}_{G_n,\epsilon}^*$.
Similarly, from Lemma \ref{Lemm:continuity-of-empirical-loss}, $\tilde{\Psi}_n$ is continuous on $\mac{F}_{G_n(M)}$, and therefore, it also attains a minimum value on the compact subset $\mac{F}_{G_n(M),\epsilon}$.
Let the point where this minimum is achieved be denoted as 
$\tilde{\bs{f}}_{n,\epsilon}^*$
.

Now, we can expand the expression as follows:
\begin{align}
\inf_{\bs{f}\in\mac{F}_{G_n(M),\epsilon}}\tilde{\Psi}_n(\bs{f})
	&=\{\tilde{\Psi}_n(\tilde{\bs{f}}_{n,\epsilon}^*)
	-\Psi(\tilde{\bs{f}}_{n,\epsilon}^*)\}
	+\{\Psi(\tilde{\bs{f}}_{n,\epsilon}^*)-\Psi(\bs{f}_{G_n,\epsilon}^*)\}
	+\{\Psi(\bs{f}_{G_n,\epsilon}^*)-\Psi(\bs{f}^*)\}
	+\Psi(\bs{f}^*).
	\label{eq:inf-tilde-Psi-n_expansion}
\end{align}
Considering that $\tilde{\bs{f}}_{n,\epsilon}^*$ belongs to $\mac{F}_{G_n(M)}$, and based on Lemmas \ref{Lemm:diff_A1} and \ref{Lemm:diff_A2}, as $n$ approaches infinity, we obtain
\begin{align*}
	 \tilde{\Psi}_n(\tilde{\bs{f}}_{n,\epsilon}^*)
	-\Psi(\tilde{\bs{f}}_{n,\epsilon}^*)
	 &\le\sup_{\bs{f}\in\mac{F}_{G_n(M)}}|\tilde{\Psi}_n(\bs{f})-\Psi_n(\bs{f})|
	+\sup_{\bs{f}\in\mac{F}_{G_n(M)}}|\Psi_n(\bs{f})-\Psi(\bs{f})|
	 \ra0\ \ \mrm{a.s.}
\end{align*}
By the optimality of $\bs{f}_{G_n,\epsilon}^*$, we have 
$\Psi(\tilde{\bs{f}}_{n,\epsilon}^*)-\Psi(\bs{f}_{G_n,\epsilon}^*)\ge0$. 
Therefore, employing \eqref{eq:limsup-tilde-Psi-n_inequality}, \eqref{eq:inf-tilde-Psi-n_expansion}, and Lemma \ref{Lemm:inequality-eps}, it can be inferred that
\begin{align}
	\liminf_{n\ra\infty}\inf_{\bs{f}\in\mac{F}_{G_n(M),\epsilon}}\tilde{\Psi}_n(\bs{f})
	>\Psi(\bs{f}^*)
	\ge\limsup_{n\ra\infty}\tilde{\Psi}_n(\tilde{\bs{f}}_n^*)
	\ \ \mrm{a.s.}
	\label{eq:Thm2-liminf-limsup}
\end{align}

In \eqref{eq:Thm2-liminf-limsup}, by the definition of the limit inferior, for sufficiently large $n$, we have
\begin{align*}
	\inf_{\bs{f}\in\mac{F}_{G_n(M),\epsilon}}\tilde{\Psi}_{n}(\bs{f})
	>\limsup_{n\ra\infty}\tilde{\Psi}_{n}(\tilde{\bs{f}}_{n}^{*})\ \
	\text{a.s.}
\end{align*}
Furthermore, by the definition of the limit superior, for sufficiently large $n$, we have
\begin{align*}
	\inf_{\bs{f}\in\mac{F}_{G_n(M),\epsilon}}\tilde{\Psi}_{n}(\bs{f})
	>\tilde{\Psi}_{n}(\tilde{\bs{f}}_{n}^{*})\ \
	\text{a.s.}
\end{align*}
Thus, if we take $\Omega_{1}$ as a subset of the sample space $\Omega$ excluding the null set, for any $\omega\in\Omega_{1}$, there exists some $n_{0}\in\mbb{N}$ such that for all $n>n_{0}$,
\begin{align}
	\inf_{\bs{f}\in\mac{F}_{G_n(M),\epsilon}}\tilde{\Psi}_{n}(\bs{f})
	>\tilde{\Psi}_{n}(\tilde{\bs{f}}_{n}^{*}).
	\label{eq:Thm2-G}
\end{align}

On the other hand, contrary to what we want to claim, assume that for some $n>n_{0}$, $d(\tilde{\bs{f}}_{n}^{*},Q)\ge\epsilon$.
In this case, by $\tilde{\bs{f}}_{n}^{*}\in\mac{F}_{G_n(M)}$ and the definition of $\tilde{\bs{f}}_n^{*}$, we have
\begin{align*}
	\inf_{\bs{f}\in\mac{F}_{G_n(M),\epsilon}}\tilde{\Psi}_{n}(\bs{f})
	=\tilde{\Psi}_{n}(\tilde{\bs{f}}_{n}^{*}).
\end{align*}
This contradicts the inequality~\eqref{eq:Thm2-G}.
Therefore,
\begin{align*}
	\forall\omega\in\Omega_{1};
	\forall\epsilon>0;
	\exists{}n_{0}\in\mbb{N};
	\forall{}n>n_{0};
	d(\tilde{\bs{f}}_{n}^{*},Q)<\epsilon.
\end{align*}
In other words, $d(\tilde{\bs{f}_{n}^{*}},Q)\ra0$ a.s. has been proved.

\section*{Acknowledgment}

This work was supported by JSPS KAKENHI Grant Numbers JP20K19756, 21H03402, 21K11787, 23H03352, and JP24K14855.

\bibliographystyle{apalike}
\bibliography{myrefs}

\clearpage

\renewcommand{\theequation}{S.\arabic{equation}}
\setcounter{equation}{0} 

\renewcommand{\thesection}{S.\arabic{section}}
\setcounter{section}{0} 

\renewcommand{\theTheo}{S.\arabic{Theo}}
\setcounter{Theo}{0} 

\renewcommand{\thefigure}{S.\arabic{figure}}
\setcounter{figure}{0} 

\renewcommand{\thetable}{S.\arabic{table}}
\setcounter{table}{0} 

\begin{center}
    {\Large $K$-means clustering for sparsely observed longitudinal data}
\end{center}
\begin{center}
    {\Large Supplementary Information}
\end{center}
\begin{center}
    {\large Michio Yamamoto, Yoshikazu Terada}
\end{center}

\section{Lemmas}
\label{SI-sec:lemmas}

All the symbols such as $X_i$, $f_{k}$, $P_{T}$, and $T_{ij}$ are used with the same meanings as in the main text.
We consider the assumptions same as those of theorems.




\begin{Lemm}
	\label{Lemm:diff_A1}
	As $n\ra\infty$, 
	\begin{align*}
		\sup_{\bs{f}\in\mac{F}_{G_n(M)}}\left|\tilde{\Psi}_{n}(\bs{f})-\Psi_{n}(\bs{f})\right|
		\ra0\ a.s.
	\end{align*}
\end{Lemm}

\begin{proof}

We observe that
\begin{align*}
	&\sup_{\bs{f}\in\mac{F}_{G_n(M)}}
	|\tilde{\Psi}_{n}(\bs{f})-\Psi_{n}(\bs{f})|\\
	&\le\frac{1}{n}\sum_{i=1}^{n}\sup_{\bs{f}\in\mac{F}_{G_n(M)}}\left|
	\min_{k=1,\dots,K}\frac{1}{N_{i}}\sum_{j=1}^{N_{i}}
	\{X_{i}(T_{ij})-f_{k}(T_{ij})\}^{2}
	-\min_{k=1,\dots,K}\|X_{i}-f_{k}\|_{P_{T}}^{2}
	\right|\\
	&\le\frac{1}{n}\sum_{i=1}^{n}\sup_{\bs{f}\in\mac{F}_{G_n(M)}}\sum_{k=1}^{K}\left|
	\frac{1}{N_{i}}\sum_{j=1}^{N_{i}}\{X_{i}(T_{ij})-f_{k}(T_{ij})\}^{2}
	-\|X_{i}-f_{k}\|_{P_{T}}^{2}
	\right|\\
	&\le\frac{1}{n}\sum_{i=1}^{n}\sup_{\bs{f}\in\mac{F}_{G_n(M)}}K\max_{k=1,\dots,K}\left|
	\frac{1}{N_{i}}\sum_{j=1}^{N_{i}}\{X_{i}(T_{ij})-f_{k}(T_{ij})\}^{2}
	-\|X_{i}-f_{k}\|_{P_{T}}^{2}
	\right|\\
	&\le{}K\frac{1}{n}\sum_{i=1}^{n}\sup_{f\in{}G_n(M)}\left|
	\frac{1}{N_{i}}\sum_{j=1}^{N_{i}}\{X_{i}(T_{ij})-f(T_{ij})\}^{2}
	-\|X_{i}-f\|_{P_{T}}^{2}
	\right|\\
	&=K\frac{1}{n}\sum_{i=1}^{n}\sup_{g_{f}\in\mac{Z}}
	\left|
	\frac{1}{N_{i}}\sum_{j=1}^{N_{i}}g_{f}(X_{i},T_{ij})
	-\mbb{E}_{T}[g_{f}(X_{i},T)]
	\right|,
\end{align*}
where the set $\mac{X}$ consists of real-valued functions on $\mac{T}$ whose supremum norm is bounded by $M$, and $\mac{Z}:=\{g_{f}:\mac{X}\times\mac{T}\ra\mbb{R}\mid{}g_{f}(x,t)=\{x(t)-f(t)\}^{2},\;f\in{}G_n(M)\}.$

In the following, as we consider each $i$ ($i=1,\dots,n$), we will temporarily omit the index $i$ from $X_{i}$ and $T_{ij}$ for simplicity of notation, representing them as $X$ and $T_j$, respectively.
From 
McDiarmid's inequality (McDiarmid, 1989),
for $\delta:0<\delta<1$, with probability at least $1-\delta$,
\mymemo{
    金森(2015)の定理2.7を使うが，結局McDiarmidの不等式を使うだけなのでMcDiarmidを引用しておく．
}
\begin{align*}
	\sup_{g_{f}\in\mac{Z}}\left|
	\frac{1}{N}\sum_{j=1}^{N}g_{f}(X,T_{j})
	-\mbb{E}_{T}[g_{f}(X,T)]
	\right|
	\le2R_{N}(\mac{Z})+C_{1}m_n\sqrt{\frac{\log(2/\delta)}{2N}}
\end{align*}
where
\begin{align*}
	R_{N}(\mac{Z})
	=\mbb{E}_{T}\left[\hat{R}_{T}(\mac{Z})\right]
	=\mbb{E}_{T}\left[\mbb{E}_{\sigma}\left(
	\sup_{g_{f}\in\mac{Z}}\frac{1}{N}\sum_{j=1}^{N}\sigma_{j}g_{f}(X,T_{j})
	\right)\right],
\end{align*}
$\sigma_{1},\dots,\sigma_{N}$ are independent random variables that take values $+1$ and $-1$ with equal probability, and $C_1$ is a constant that is independent of $n$.
In the above, we define the empirical Rademacher complexity as follows:
\begin{align*}
	\hat{R}_{T}(\mac{Z})
	&=\mbb{E}_{\sigma}\left(
	\sup_{g_{f}\in\mac{Z}}\frac{1}{N}\sum_{j=1}^{N}\sigma_{j}g_{f}(X,T_{j})
	\right).
\end{align*}
Then,
\begin{align*}
	\hat{R}_{T}(\mac{Z})
	&=\mbb{E}_{\sigma}\left[\sup_{f\in{}G_n(M)}\frac{1}{N}\sum_{j=1}^{N}\sigma_{j}(X(T_j)-f(T_j))^2\right]\notag\\
	&\le\mbb{E}_{\sigma}\left[\frac{1}{N}\sum_{j=1}^{N}\sigma_{j}\{X(T_{j})\}^{2}
	+\sup_{f\in{}G_n(M)}\frac{1}{N}\sum_{j=1}^{N}2\sigma_{j}(-1)X(T_{j})f(T_j)
	+\sup_{f\in{}G_n(M)}\frac{1}{N}\sum_{j=1}^{N}\sigma_{j}\{f(T_{j})\}^{2}\right]\notag\\
	&=2\mbb{E}_{\sigma}\left[\sup_{f\in{}G_n(M)}\frac{1}{N}\sum_{j=1}^{N}\sigma_{j}X(T_{j})f(T_{j})\right]
	+\mbb{E}_{\sigma}\left[\sup_{f\in{}G_n(M)}\frac{1}{N}\sum_{j=1}^{N}\sigma_{j}\{f(T_{j})\}^{2}\right].
\end{align*}

First, let's consider the first term on the right-hand side of the last line.
Let $\|\cdot\|_2$ denote the Euclidean norm.
By the Cauchy–Schwarz inequality and the assumption $\sup_{t}|X(t)|\le M$ a.s.,
\ba
&\mbb{E}_{\sigma}\left[\sup_{f\in{}G_n(M)}\frac{1}{N}\sum_{j=1}^{N}\sigma_{j}X(T_{j})f(T_{j})\right]
=
\mbb{E}_{\sigma}\left[
\sup_{\bm{\beta}:\|\bm{\beta}\|_2 \le M} \bigg\{\frac{1}{N}\sum_{j=1}^{N}\sigma_{j}X(T_j)\bm{\phi}(T_j)\bigg\}^\top\bm{\beta}
\right]\\
&\le
\mbb{E}_{\sigma}\left[\sup_{\bm{\beta}:\|\bm{\beta}\|_2 \le M} \|\bm{\beta}\|_2 \left\|\frac{1}{N}\sum_{j=1}^{N}\sigma_{j}X(T_j)\bm{\phi}(T_j) \right\|_2 \right]
\le
M\mbb{E}_{\sigma}\left[\left\| \frac{1}{N}\sum_{j=1}^{N}\sigma_{j}X(T_j)\bm{\phi}(T_j) \right\|_2 \right]\\
&\le
M\sqrt{\mbb{E}_{\sigma}\left[\left\| \frac{1}{N}\sum_{j=1}^{N}\sigma_{j}X(T_j)\bm{\phi}(T_j) \right\|_2^2 \right]}
=
M\sqrt{\frac{1}{N^2}\sum_{j=1}^{N}\mbb{E}_{\sigma}\big[X^2(T_j)\|\bm{\phi}(T_j)\|_2^2 \big]}
\le
C_\phi M^2\sqrt{\frac{m_n}{N}}.
\ea

We recall that $|f(t)| \le C_\phi M \sqrt{m_n}$ for all $f \in G_n(M)$.
For the second term, similar to the first term, 
the comparison inequality \citep[Theorem 4.12]{ledoux2013probability}
implies
\ba
    &\mbb{E}_{\sigma}\left[\sup_{f\in{}G_n(M)}\frac{1}{N}\sum_{j=1}^{N}\sigma_{j}\{f(T_{j})\}^{2}\right]
    \le 
    4C_\phi^2 M^2 m_n\mbb{E}_{\sigma}\left[\sup_{f\in{}G_n(M)}\left| \frac{1}{N}\sum_{j=1}^{N}\sigma_{j}\left\{\frac{f(T_{j})}{2C_\phi M \sqrt{m_n}}\right\}^{2}\right|\right]\\
    &\le
    8C_\phi^2 M^2 m_n\mbb{E}_{\sigma}\left[\sup_{f\in{}G_n(M)}\left| \frac{1}{N}\sum_{j=1}^{N}\sigma_{j}\frac{f(T_{j})}{2C_\phi M \sqrt{m_n}}\right|\right]
    =
    4C_\phi M \sqrt{m_n}\;\mbb{E}_{\sigma}\left[\sup_{f\in{}G_n(M)}\left| \frac{1}{N}\sum_{j=1}^{N}\sigma_{j}f(T_{j})\right|\right]\\
    &\le 4C_\phi^2 M^2 \sqrt{\frac{m_n^2}{N}}.
    \ea

Considering this simultaneously for all $i=1,\dots,n$, we have
\begin{align*}
	&\mrm{P}\left[
	K\frac{1}{n}\sum_{i=1}^{n}\sup_{g_{f}\in\mac{Z}}\left|
	\frac{1}{N_{i}}\sum_{j=1}^{N_{i}}g_{f}(X_{i},T_{ij})-\mbb{E}_{T}\left[g_{f}(X,T)\right]
	\right|
	\le{}K\frac{1}{n}\sum_{i=1}^{n}
	\left(C_{2}\sqrt{\frac{m_{n}^{2}}{N_{i}}}+C_{1}m_n\sqrt{\frac{\log(2\delta^{-1})}{2N_{i}}}
	\right)
	\right]\\
	&\ge(1-\delta)^{n}
\end{align*}
where $C_2$ is a constant independent of $n$.
By Bernoulli's inequality, we have $(1-\delta)^n \ge 1- n\delta$.
Let $D_{n}:=n^{-1}\sum_{i=1}^{n}N_i^{-1/2}$.
Then, 
\begin{align*}
	K\frac{1}{n}\sum_{i=1}^{n}\left(C_{2}\sqrt{\frac{m_{n}^{2}}{N_{i}}}+C_{1}m_n\sqrt{\frac{\log(2/\delta)}{2N_{i}}}\right)
	&\le C_{3}Km_n\left(1+\sqrt{\log(2/\delta)}\right)D_{n}\\
	&\le{}C_{4}Km_n\sqrt{1+\log(2/\delta)}D_{n}
	=:\epsilon,
\end{align*}
where $C_{4}:=C_{3}\sqrt{2}$.

Here, expressing $\delta$ in terms of $\epsilon$, we have
\ba
\delta = 2\exp\left( - \frac{\epsilon^2}{C_4^2K^2D_n^2m_n^2} + 1 \right).
\ea
When $\mrm{P}(A\ge{}a)\le\delta n$, if $a\le\epsilon$, then $\mrm{P}(A\ge\epsilon)\le\delta n$, and thus,
\ba
	\mrm{P}\left[
	\sup_{\bs{f}\in\mac{F}_{G_n(M)}}|\tilde{\Psi}_{n}(\bs{f})-\Psi_{n}(\bs{f})|
	\ge\epsilon
	\right]
	&\le{}2n\exp\left(
		-\frac{\epsilon^2}{C_4^2K^2D_n^2m_n^2} + 1 
	\right)\\
	&=\exp\left(-C_{5}\frac{\epsilon^{2}}{K^{2}D_{n}^{2}m_n^2}+1+\log2+\log n\right)
\ea
where $C_{5}:=1/C_{4}^{2}$.

\mymemo{
    XXX This part should be reconsidered as follows. The quantity $a_n$ is not necessary. XXX
    For any $\epsilon>0$, we require that
    \[
    C \times \frac{\epsilon}{D_n^2m_n^2} - \log(n) - 1
    =
    \left( C \times \frac{\epsilon}{D_n^2m_n^2\log n} - 1 - \frac{1+\log2}{\log n} \right)
    \log n
    \rightarrow \infty
    \]
    If $D_n^2m_n^2\log n \rightarrow 0$, we can ensure the convergence.
}

From the Borel-Cantelli lemma, for $n\ra\infty$, if $D_n^2m_n^2\log n \rightarrow 0$, then
\begin{align*}
	\sup_{\bs{f}\in\mac{F}_{G_n(M)}}|\tilde{\Psi}_{n}(\bs{f})-\Psi_{n}(\bs{f})|
	\ra0\ \
	\text{a.s.}
\end{align*}
is obtained.

 \end{proof}


The following lemma is obtained by the proof of Corollary 3.1 of \citet{biau2008performance}.
For the sake of completeness, we provide the proof of the lemma.
\begin{Lemm}
	\label{Lemm:diff_A2}
	As $n\ra\infty$, 
	\begin{align*}
		\sup_{\bs{f}\in\mac{F}_{G_n(M)}}\left|\Psi_{n}(\bs{f})-\Psi(\bs{f})\right|
		\ra0\ a.s.
	\end{align*}
\end{Lemm}

\begin{proof}
From Lemma 4.3 of Biau et al. (2008), we have
\begin{align*}
	\mbb{E}\sup_{\bs{f}\in\mac{F}_{G_n(M)}}(\Psi_{n}(\bs{f})-\Psi(\bs{f}))
	\le 8K\frac{M^{2}}{\sqrt{n}}.
\end{align*}
Additionally, we can obtain the same upper bound by replacing with $\Psi(\bs{f})-\Psi_{n}(\bs{f})$, and thus,
\begin{align*}
	\mbb{E}\sup_{\bs{f}\in\mac{F}_{G_n(M)}}|\Psi_{n}(\bs{f})-\Psi(\bs{f})|
	\le8K\frac{M^{2}}{\sqrt{n}}.
\end{align*}

We will consider the set $\mac{G}$ of functions defined as
\begin{align*}
	\mac{G}
	=\left\{g:\mac{X}\ra\mbb{R}\ \Big|\ {}g(x)=\min_{k=1,\dots,K}\|x-f_{k}\|_{P_{T}}^{2},
	\{f_{1},\dots,f_{K}\}\in\mac{F}_{G(M)}\right\}.
\end{align*}
Note that for $g\in\mac{G}$,
\begin{align*}
	|g(x)|
	=\min_{k}\left|
	\|x\|_{P_{T}}^{2}+\|f_{k}\|_{P_{T}}^{2}-2\lan{}x,f_{k}\ran_{P_{T}}
	\right|
	\le\|x\|_{P_{T}}^{2}+\|f_{1}\|_{P_{T}}^{2}+2\left|\lan{}x,f_{1}\ran_{P_{T}}\right|
	\le 4M^{2}.
\end{align*}
Furthermore, for $x_{1},\dots,x_{n}\in\mac{X}$, define a function $A(x_{1},\dots,x_{n})$ as follows:
\begin{align*}
	A(x_{1},\dots,x_{n})
	&=\sup_{\bs{f}\in\mac{F}_{G_n(M)}}\left\{
            \Psi_{n}(\bs{f})-\Psi(\bs{f})
        \right\}
	=\sup_{g\in\mac{G}}
        \left\{
            \frac{1}{n}\sum_{i=1}^{n}g(x_{i})-\mbb{E}g(X)
        \right\}.
\end{align*}
Then, for any $x'\in\mac{X}$, the following inequality holds:
\begin{align*}
	&\left|A(x_{1},\dots,x_{n-1},x_{n})
	-A(x_{1},\dots,x_{n-1},x')
        \right|\\
	&=\left|\sup_{g\in\mac{G}}\inf_{f\in\mac{G}}\left\{
	\frac{1}{n}\sum_{i=1}^{n}g(x_{i})-\mbb{E}g(X)-\frac{1}{n}\sum_{i=1}^{n-1}f(x_{i})-\frac{1}{n}f(x')+\mbb{E}f(X)
	\right\}\right|\\
	&\le\left|\sup_{g\in\mac{G}}\left\{
	\frac{1}{n}\sum_{i=1}^{n}g(x_{i})-\mbb{E}g(X)-\frac{1}{n}\sum_{i=1}^{n-1}g(x_{i})-\frac{1}{n}g(x')+\mbb{E}g(X)
	\right\}\right|\\
	&\le\sup_{g\in\mac{G}}\frac{|g(x_{n})|+|g(x')|}{n}
	\le\frac{8M^{2}}{n}.
\end{align*}
\mymemo{
	【3行目の不等号は$f$を$g\in\mac{G}$に置き換えた方が等しいか大きくなる事から成立】
}
Therefore, from McDiarmid's inequality, for any $\delta\in(0,1)$, we have
\begin{align*}
\mrm{P}\left(
	A(X_{1},\dots,X_{n})-\mbb{E}A(X_{1},\dots,X_{n})
	\le8M^{2}\sqrt{\frac{\log(1/\delta)}{2n}}
	\right)
	\ge1-\delta.
\end{align*}
Similarly, for $A'(x_{1},\dots,x_{n}) = \sup_{\bs{f}\in\mac{F}_{G_n(M)}}\left\{\Psi(\bs{f})-\Psi_{n}(\bs{f})\right\}$,
we have the same inequality.
Thus, with probability at least $1-2\delta$, for sufficiently large $n$, we get
\mymemo{
    $P(X\le a)\ge1-\delta$ならば$P(X>a)\le\delta$なので，$P(X>a\vee Y>a)\le P(X>a)+P(Y>a)\le2\delta$である．
    ゆえに，$P(X\le a\wedge Y\le a)=1-P(X>a\vee Y>a)\ge1-2\delta$となる．
}
\begin{align*}
	&\sup_{\bs{f}\in\mac{F}_{G_n(M)}}|\Psi_{n}(\bs{f})-\Psi(\bs{f})|
    = \sup_{\bs{f}\in\mac{F}_{G_n(M)}}\{\Psi_{n}(\bs{f})-\Psi(\bs{f})\} \vee \sup_{\bs{f}\in\mac{F}_{G_n(M)}}\{\Psi(\bs{f})-\Psi_{n}(\bs{f})\}\\
	&\le\mbb{E}\left[\sup_{\bs{f}\in\mac{F}_{G_n(M)}}\{\Psi_{n}(\bs{f})-\Psi(\bs{f})\}\right] \vee \mbb{E}\left[\sup_{\bs{f}\in\mac{F}_{G_n(M)}}\{\Psi(\bs{f})-\Psi_n(\bs{f})\}\right]
	+8M^{2}\sqrt{\frac{\log(1/\delta)}{2n}}\\
	&\le\frac{8KM^{2}}{\sqrt{n}}
	+\frac{4\sqrt{2}M^{2}}{\sqrt{n}}\log(1/\delta).
\end{align*}

Therefore, we have
\begin{align*}
	\mrm{P}\left(
	\sup_{\bs{f}\in\mac{F}_{G_n(M)}}|\Psi_{n}(\bs{f})-\Psi(\bs{f})|>\epsilon
	\right)
	\le
	2\exp\left(
	-\frac{\sqrt{n}}{4\sqrt{2}M^{2}}\epsilon
	+\sqrt{2}K
	\right).
\end{align*}
Consequently, from the Borel-Cantelli lemma, as $n\ra\infty$, $\sup_{\bs{f}\in\mac{F}_{G_n(M)}}|\Psi_{n}(\bs{f})-\Psi(\bs{f})|$ converges to zero alomost surely.

\end{proof}





\vspace{10pt}
\begin{Lemm}
	\label{Lemm:diff_A3}
	As $n\ra\infty$, 
	$\Psi(\bs{f}_{G_n}^{*})$ converges to $\Psi(\bs{f}^{*})$.
\end{Lemm}

\begin{proof}
From the projection theorem, there exists $\tilde{f}_{k}\in{}G_n(M)$ such that $\|\tilde{f}_{k}-f_{k}^{*}\|_{P_{T}}=\inf_{g\in{}G_n(M_1)}\|g-f_{k}^{*}\|_{P_{T}}$ with $\bs{f}^*=(f_1^*,\dots,f_K^*)^\top$.
Now, since $\{\phi_j\}_{j=1,\dots,\infty}$ is a CONS in $L_2(P_T)$, 
$f_k^*$ can be expressed as $f_k^*=\sum_{j=1}^{\infty}a_{kj}^*\phi_j$.
Given the orthogonality of $\phi_l$, $\tilde{f}_k$ can be represented as $\tilde{f}_k=\sum_{j=1}^{m_n}a_{kj}^*\phi_j$.
Thus, for $k=1,\dots,K$,
$\|\tilde{f}_k-f_k^*\|_{P_T}^2 = \sum_{j=m_n+1}^{\infty}{a_{kj}^*}^2\rightarrow 0$ as $m_n \rightarrow \infty$.
By the continuity of the loss function $\Psi$ (Lemma~3),
we obtain
$
\Psi(\bs{f}_{G_n}^{*})-\Psi(\bs{f}^{*})
\le
\Psi(\tilde{\bs{f}})-\Psi(\bs{f}^{*}) \rightarrow 0
$
as
$m_n \rightarrow \infty.$

\end{proof}

\section{Proofs}
\label{SI-sec:proofs}

\subsection{Proof of Lemma \ref{Lemm:continuity-of-loss}}

We can express $\Psi(\bs{f}) =\mbb{E}\min_{k}\|X-f_{k}\|_{P_{T}}^{2}=\int\min_{k}\|x-f_{k}\|_{P_{T}}^{2}dP_{X}(x)$.
For any $\delta > 0$, let $\bs{f},\bs{g}$ are elements in $\mac{F}_{(M)}$ such that $d_{H}(\bs{f},\bs{g})<\delta$.
From the definition of $d_H$, for each $g_{\ell}\in\bs{g}$, there exists $f_{k}(g_{\ell})\in\bs{f}$ with $\|g_{\ell}-f_{k}(g_{\ell})\|_{P_{T}}<\delta$.
Then, we have
\begin{align*}
	\Psi(\bs{f})-\Psi(\bs{g})
	&=\int\left[
	\min_{k}\|x-f_{k}\|_{P_{T}}^{2}-\min_{\ell}\|x-g_{\ell}\|_{P_{T}}^{2}
	\right]dP_{X}(x)\\
	&\le\int\left[
	\min_{\ell'}\|x-f_{k}(g_{\ell'})\|_{P_{T}}^{2}-\min_{\ell}\|x-g_{\ell}\|_{P_{T}}^{2}
	\right]dP_{X}(x)\\
	&\le\int\max_{\ell}\left[
	\|x-f_{k}(g_{\ell})\|_{P_{T}}^{2}-\|x-g_{\ell}\|_{P_{T}}^{2}
	\right]dP_{X}(x)\\
	&\le\int\max_{\ell}\left[
	(\|x-g_{\ell}\|_{P_{T}}+\delta)^{2}-\|x-g_{\ell}\|_{P_{T}}^{2}
	\right]dP_{X}(x)
	\\
	&\le\sum_{\ell=1}^{K}\int\left[
	(\|x-g_{\ell}\|_{P_{T}}+\delta)^{2}-\|x-g_{\ell}\|_{P_{T}}^{2}
	\right]dP_{X}(x)\\
        &\le K\delta^2+4KM\delta.
\end{align*}

Taking $\delta$ sufficiently small, the upper bound of the above equations can be made arbitrarily small.
Similarly, for $\Psi(\bs{g})-\Psi(\bs{f})$, interchanging the roles of $\bs{f}$ and $\bs{g}$ yields the same result.
Therefore, the uniform continuity of $\Psi(\cdot)$ on $\mac{F}_{(M)}$ has been proven.

\subsection{Proof of Lemma \ref{Lemm:continuity-of-empirical-loss}}
\label{subsec:proof-lemma-continuity-empirical-loss}

Suppose that we have a sample $\{T_{ij},X_i(T_{ij});i=1,\dots,n,j=1,\dots,N_i\}$.
At this time, let us express $\tilde{\Psi}_n$ using the semi-norm $\|\cdot\|_{N_i}$ defined below.
For each $i=1,\dots,n$, we define the semi-inner product $\lan \cdot, \star \ran_{N_i}$ on $L_2(P_T)$ as
\ba
	\lan f,g\ran_{N_i}
	:=\sum_{j=1}^{N_i}\frac{1}{N_i}f(T_{ij})g(T_{ij})
	\ \ \text{for}\ f,g\in L_2(P_T).
\ea
\mymemo{
	It can be easily verified that this function $\lan\cdot,\cdot\ran$ satisfies the definition of a semi-inner product.
}
Note that, for $f\in L_2(P_T)$, $f$ is not necessarily zero even if $\lan f,f \ran_{N_i}=0$, and thus, $\lan\cdot,\star\ran_{N_i}$ is not an inner product.
This semi-inner product defines the semi-norm $\|\cdot\|_{N_i}:=\lan\cdot,\cdot\ran_{N_i}^{1/2}$.
Using this semi-norm, the empirical loss function can be reexpressed as
\begin{align*}
	\tilde{\Psi}_n(\bs{f})
	=\frac{1}{n}\sum_{i=1}^n\min_{k=1,\dots,K}\|X_i-f_k\|_{N_i}^2,
 \quad \bs{f}=(f_1,\dots,f_K)\in\mac{F}.
\end{align*}

For any $f,g\in G_n(M)$, let $f=\bs{\beta}^{\top}\bs{\phi}$, $g=\bs{\gamma}^{\top}\bs{\phi}$.
Then, 
we have $\|\bs{\beta}-\bs{\gamma}\|_2=\|f-g\|_{P_T}$.
Therefore, for any $\delta>0$, if $\|f-g\|_{P_T}\le\delta$, then using the assumption $\|\phi_\ell\|_{\infty}\le C_\phi$, we have
\begin{align*}
	\|f-g\|_{N_i}^2
	&=\sum_{j=1}^{N_i}\frac{1}{N_i}(f(T_{ij})-g(T_{ij}))^2
	=\sum_{j=1}^{N_i}\frac{1}{N_i}\{(\bs{\beta}-\bs{\gamma})^{\top}\bs{\phi}(T_{ij})\}^2\\
	&\le\sum_{j=1}^{N_i}\frac{1}{N_i}\|\bs{\beta}-\bs{\gamma}\|_2^2\cdot\|\bs{\phi}(T_{ij})\|_2^2
	=\sum_{j=1}^{N_i}\frac{1}{N_i}\delta^2m_n C_\phi^2
	=m_n C_\phi^2\delta^2.
\end{align*}

For any $\delta > 0$, choose $\bs{f},\bs{g}\in\mac{F}_{G(M)}$ such that $d_{H}(\bs{f},\bs{g})<\delta$.
From the definition of $d_H$, for each $g_{\ell}\in\bs{g}$, there exists $f(g_{\ell})\in\bs{f}$ satisfying $\|g_{\ell}-f(g_{\ell})\|_{P_{T}}<\delta$.
In this case, we obtain
\begin{align*}
	\tilde{\Psi}_n(\bs{f})-\tilde{\Psi}_n(\bs{g})
	&=\frac{1}{n}\sum_{i=1}^n\left[
	\min_{k}\|X_i-f_k\|_{N_i}^2-\min_{\ell}\|X_i-g_\ell\|_{N_i}^2
	\right]\\
	&\le\frac{1}{n}\sum_{i=1}^n\left[
	\min_{\ell'}\|X_i-f(g_{\ell'})\|_{N_i}^2-\min_{\ell}\|X_i-g_\ell\|_{N_i}^2
	\right]\\
	&\le\frac{1}{n}\sum_{i=1}^n\max_{\ell}\left[
	\|X_i-f(g_{\ell})\|_{N_i}^2-\|X_i-g_\ell\|_{N_i}^2
	\right]\\
	&\le\frac{1}{n}\sum_{i=1}^n\max_{\ell}\left[
	(\|X_i-g_\ell\|_{N_i}+\sqrt{m_n}C_\phi\delta)^2-\|X_i-g_\ell\|_{N_i}^2
	\right]
	\\
	&\le\frac{1}{n}\sum_{i=1}^n\sum_{\ell=1}^K\left[
	(\|X_i-g_\ell\|_{N_i}+\sqrt{m_n}C_\phi\delta)^2-\|X_i-g_\ell\|_{N_i}^2
	\right].
\end{align*}


Thus, if $\delta$ is taken sufficiently small, the upper bound of $\tilde{\Psi}_n(\bs{f})-\tilde{\Psi}_n(\bs{g})$ can be made arbitrarily small.
Similarly, for $\tilde{\Psi}_n(\bs{g})-\tilde{\Psi}_n(\bs{f})$, interchanging the roles of $\bs{f}$ and $\bs{g}$ yields the same result.
This demonstrates the uniform continuity of $\tilde{\Psi}_n(\cdot)$ on $\mac{F}_{G_n(M)}$.

\subsection{Proof of Lemma \ref{lemm:lim-tilde-Psi-n}}

We have
\begin{align*}
	\tilde{\Psi}_n(\bs{f}_{G_n}^*)
	&=\{\tilde{\Psi}_n(\bs{f}_{G_n}^*)-\Psi_n(\bs{f}_{G_n}^*)\}
	+\{\Psi_n(\bs{f}_{G_n}^*)-\Psi(\bs{f}_{G_n}^*)\}
	+\{\Psi(\bs{f}_{G_n}^*)-\Psi(\bs{f}^*)\}
	+\Psi(\bs{f}^*).
\end{align*}

Here, we note that for $f \in G_n$,
\[
\|x - f\|_{P_T}^2 
=
\|x - P_{G_n}x + P_{G_n}x -  f\|_{P_T}^2 
=
\|x - P_{G_n}x\|^2 + \|P_{G_n}x -  f\|_{P_T}^2,
\]
where $P_{G_n}$ is the projection operator onto the subspace $G_n$.
Thus, 
\[
\bs{f}_{G_n}^* \in 
\mathop{\arg\min}_{\bs{f} \in \mac{F}_{G_n}}
\mathbb{E}\left[ \min_{1\le k\le K}\|P_{G_n}X -  f_k\|_{P_T}^2 \right].
\]
From the so-called centroid condition (see, for example, Section 6.2 of \citet{gersho2012vector}), $\bs{f}_{G_n}^{*}\in\mac{F}_{G_n(M)}$
since $\|P_{G_n}X\|_{P_T}\le \|X\|_{P_T} \le M$ a.s.
Then, from Lemma \ref{Lemm:diff_A1} and Lemma \ref{Lemm:diff_A2}, as $n\ra\infty$, we have
\begin{align*}
	\tilde{\Psi}_n(\bs{f}_{G_n}^*)-\Psi_n(\bs{f}_{G_n}^*)\ra0\ \ 
	\text{a.s.} \;\text{ and }\;
	\Psi_n(\bs{f}_{G_n}^*)-\Psi(\bs{f}_{G_n}^*)\ra0\ \ 
	\text{a.s.},
\end{align*}
respectively.
Therefore, from Lemma \ref{Lemm:diff_A3}, we obtain
$\tilde{\Psi}_n(\bs{f}_{G_n}^*)\ra\Psi(\bs{f}^*)$ a.s. 
as $n\ra\infty$.

\subsection{Proof of Lemma \ref{Lemm:inequality-eps}}

Similar to $Q(\mac{F}_{G_n}\backslash B^o(Q,\epsilon))$, we define
\begin{align*}
    Q(\mac{F}\backslash B^o(Q,\epsilon))
	 & :=\left\{\bs{f}\in\mac{F}\;\big|\;\inf_{\bs{f}'\in\mac{F}\backslash B^o(Q,\epsilon)}\Psi(\bs{f}')=\Psi(\bs{f})
  \right\}.
\end{align*}
According to Lemma 4.4 of \citet{levrard2015nonasymptotic}, for any $n$, elements achieving the two lower bounds denoted in the definitions of $Q(\mac{F}\backslash B^o(Q,\epsilon))$ and $Q(\mac{F}_{G_n}\backslash B^o(Q,\epsilon))$ exist in $\mac{F}_{(M+\epsilon)}\backslash B^o(Q,\epsilon)$ and $\mac{F}_{G_n(M+\epsilon)}\backslash B^o(Q,\epsilon)$, respectively.
Since $\mac{F}_{G_n}\subset\mac{F}$, we have 
$
	\lim_{n\ra\infty}\Psi(\bs{f}_{G_n,\epsilon}^*)\ge\Psi(\bs{f}_{\epsilon}^*)
$
for any $\bs{f}_{\epsilon}^*\in Q(\mac{F}\backslash B^o(Q,\epsilon))$.
Therefore, we obtain
$
    \lim_{n\ra\infty}\Psi(\bs{f}_{G_n,\epsilon}^*)
    \ge\Psi(\bs{f}_{\epsilon}^*)
    >\Psi(\bs{f}^*)
$, which demonstrates the lemma.

\clearpage

\section{Details on experiments}
\label{SI-sec:details-experiments}



\subsection{Cluster centers estimated by FKM with Fourier basis functions in the artificial experiment}
\label{SI-subsec:cluster-centers-Fourier}

\begin{figure}[!h]
	\begin{center}
		\renewcommand{\arraystretch}{1}
		\vspace{0.5cm}
		\begin{tabular}{cccc}
			$(n,N_{tp})=(50,3)$ & $(n,N_{tp})=(100,3)$ & $(n,N_{tp})=(200,3)$ & $(n,N_{tp})=(1000,3)$\\
			\includegraphics[width=4cm]{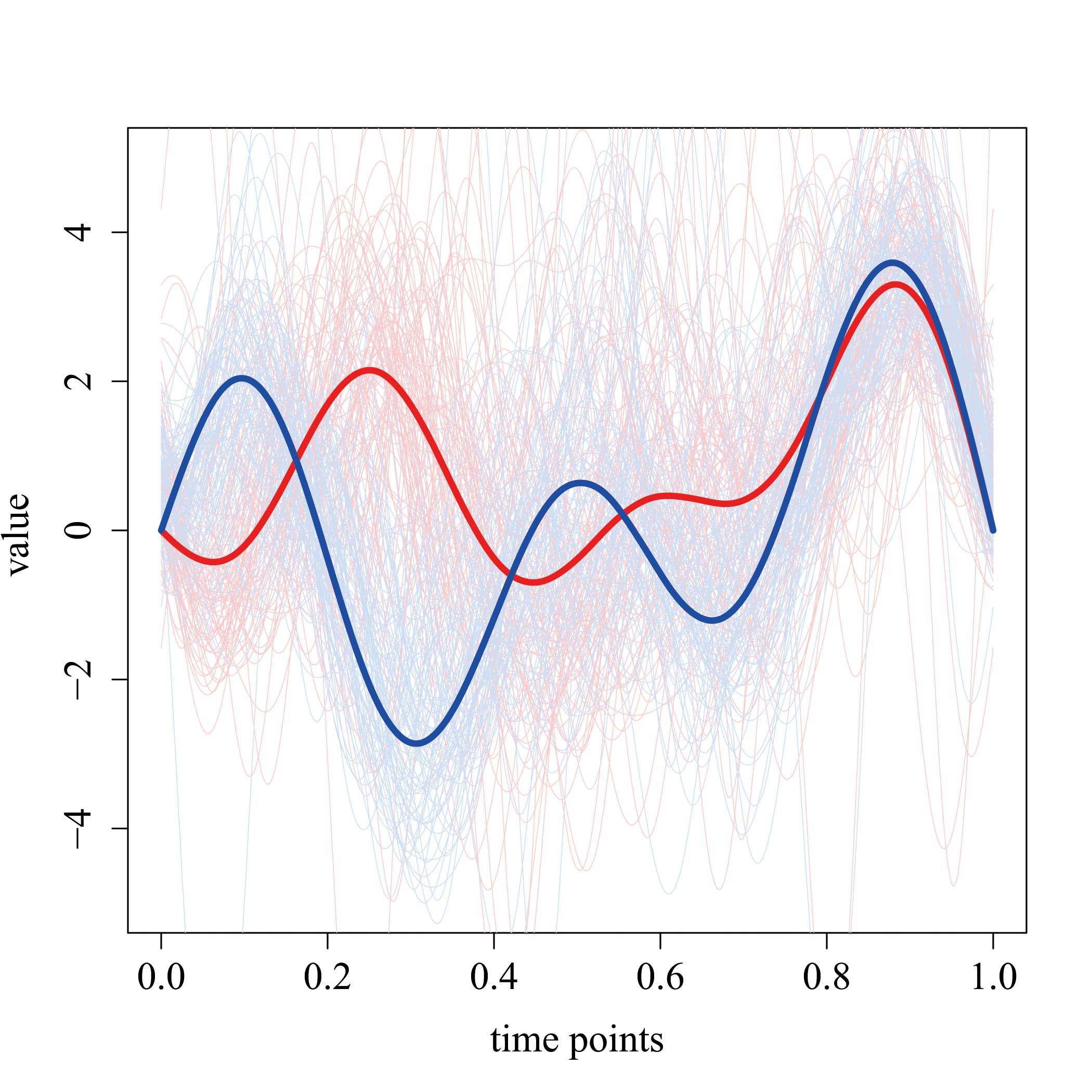} & 	
			\includegraphics[width=4cm]{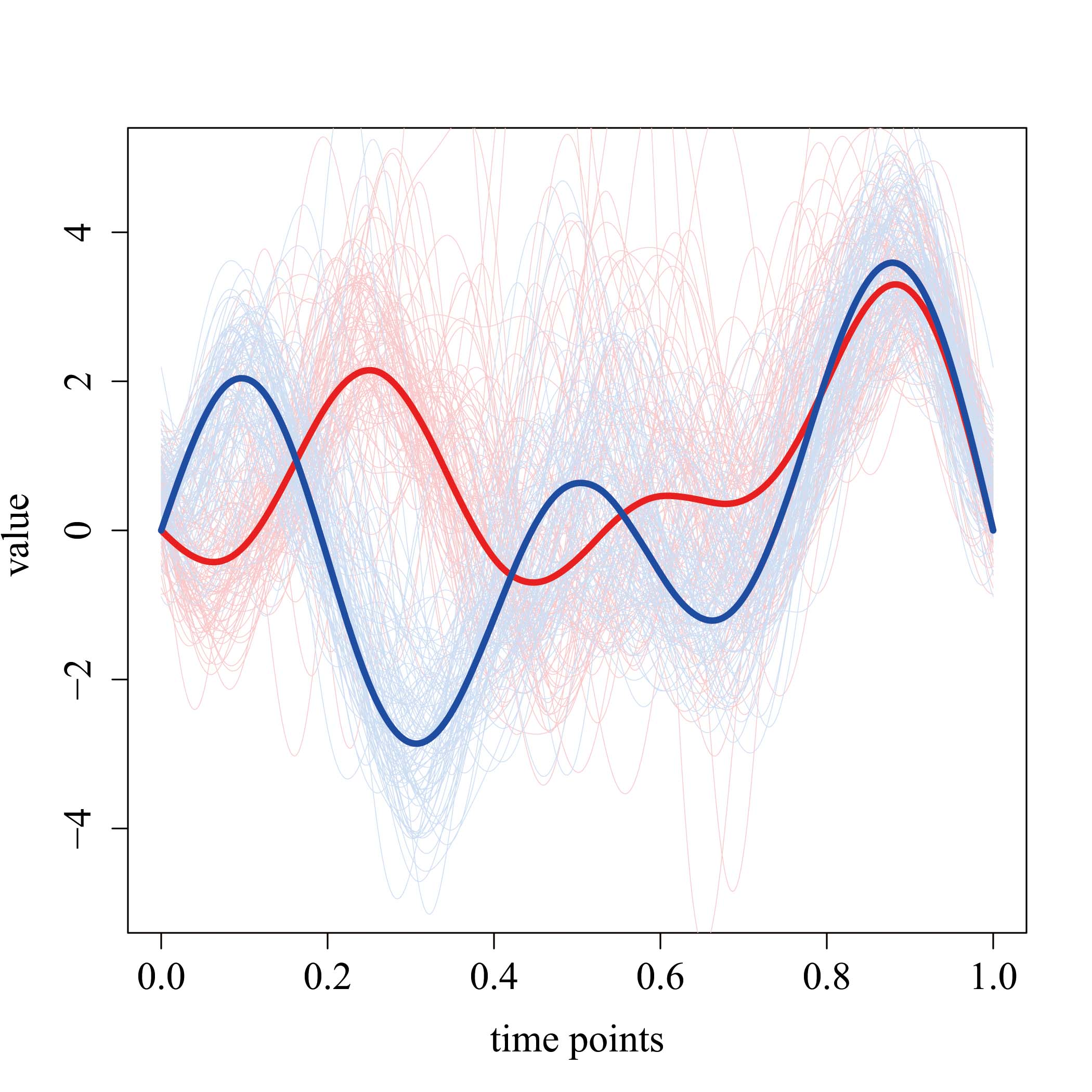} & 	
			\includegraphics[width=4cm]{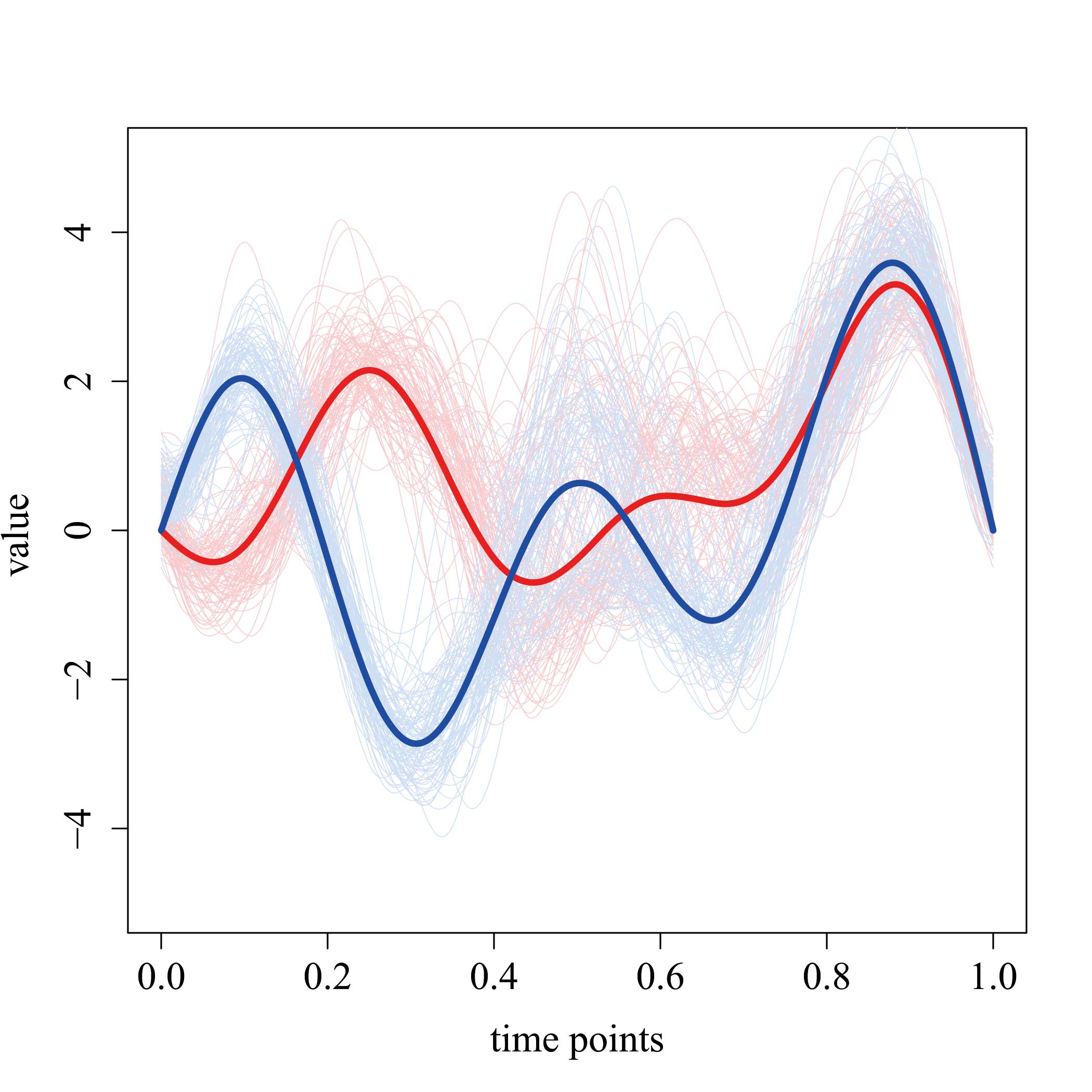} & 	
			\includegraphics[width=4cm]{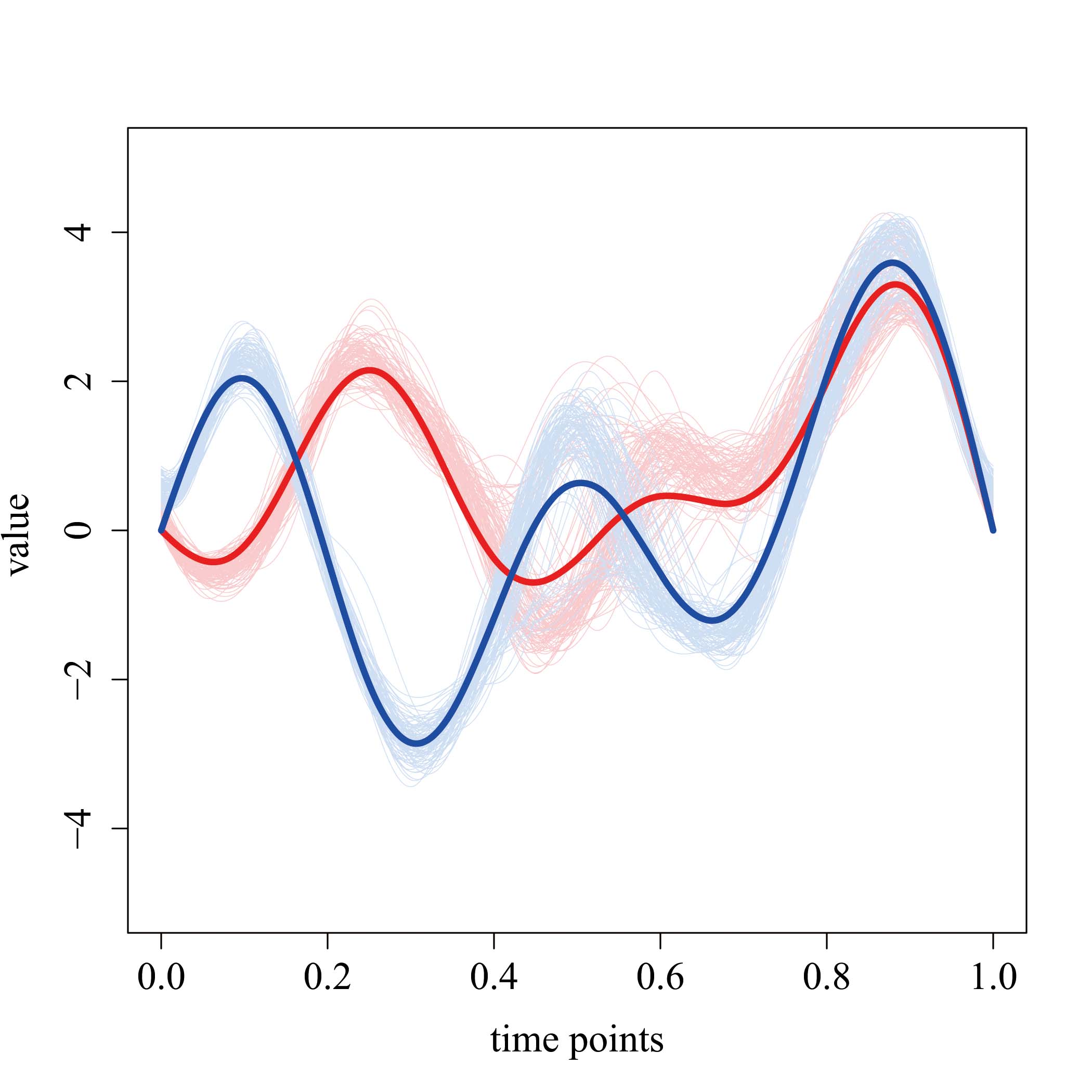}\\\\
			$(n,N_{tp})=(50,5)$ & $(n,N_{tp})=(100,5)$ & $(n,N_{tp})=(200,5)$ & $(n,N_{tp})=(1000,5)$\\
			\includegraphics[width=4cm]{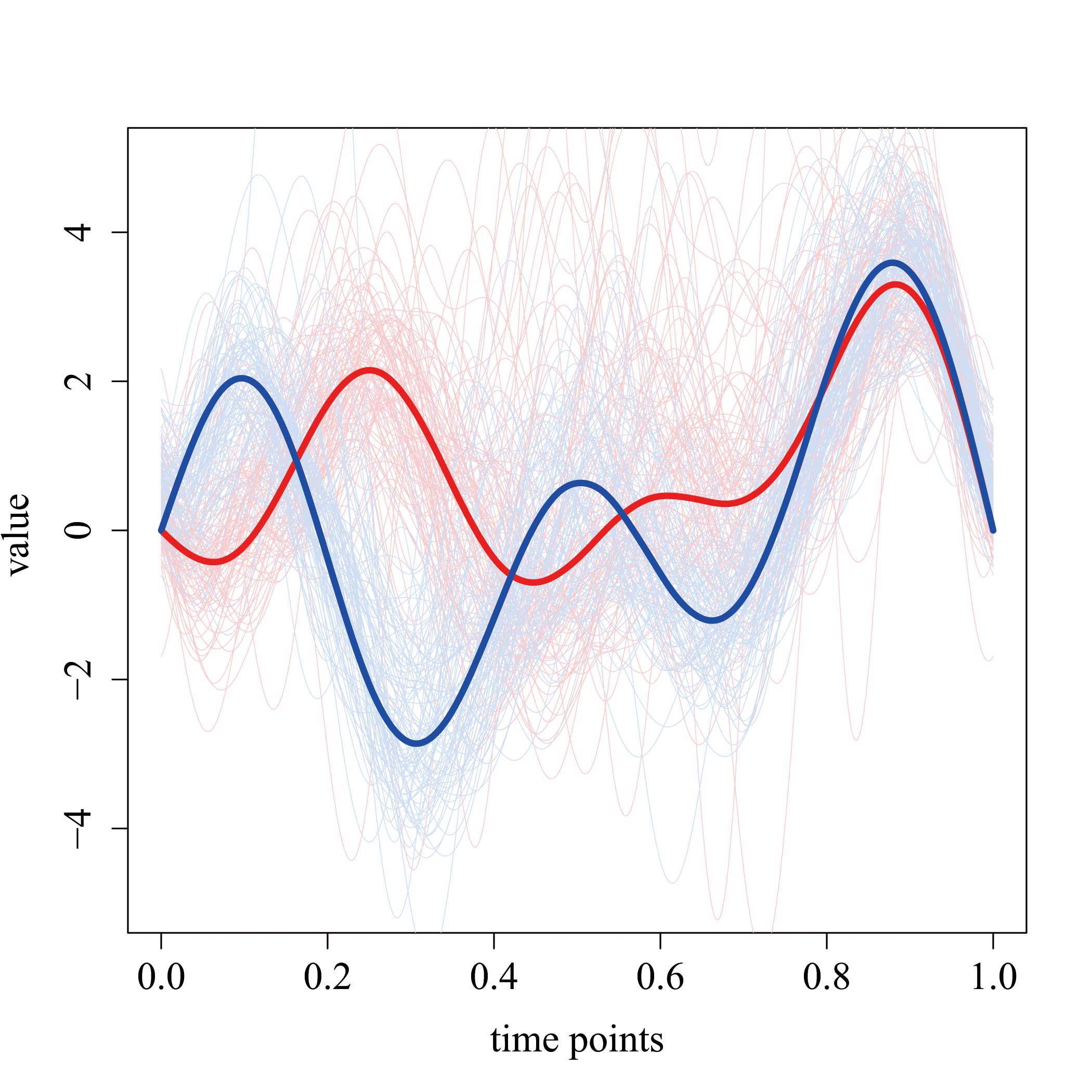} & 	
			\includegraphics[width=4cm]{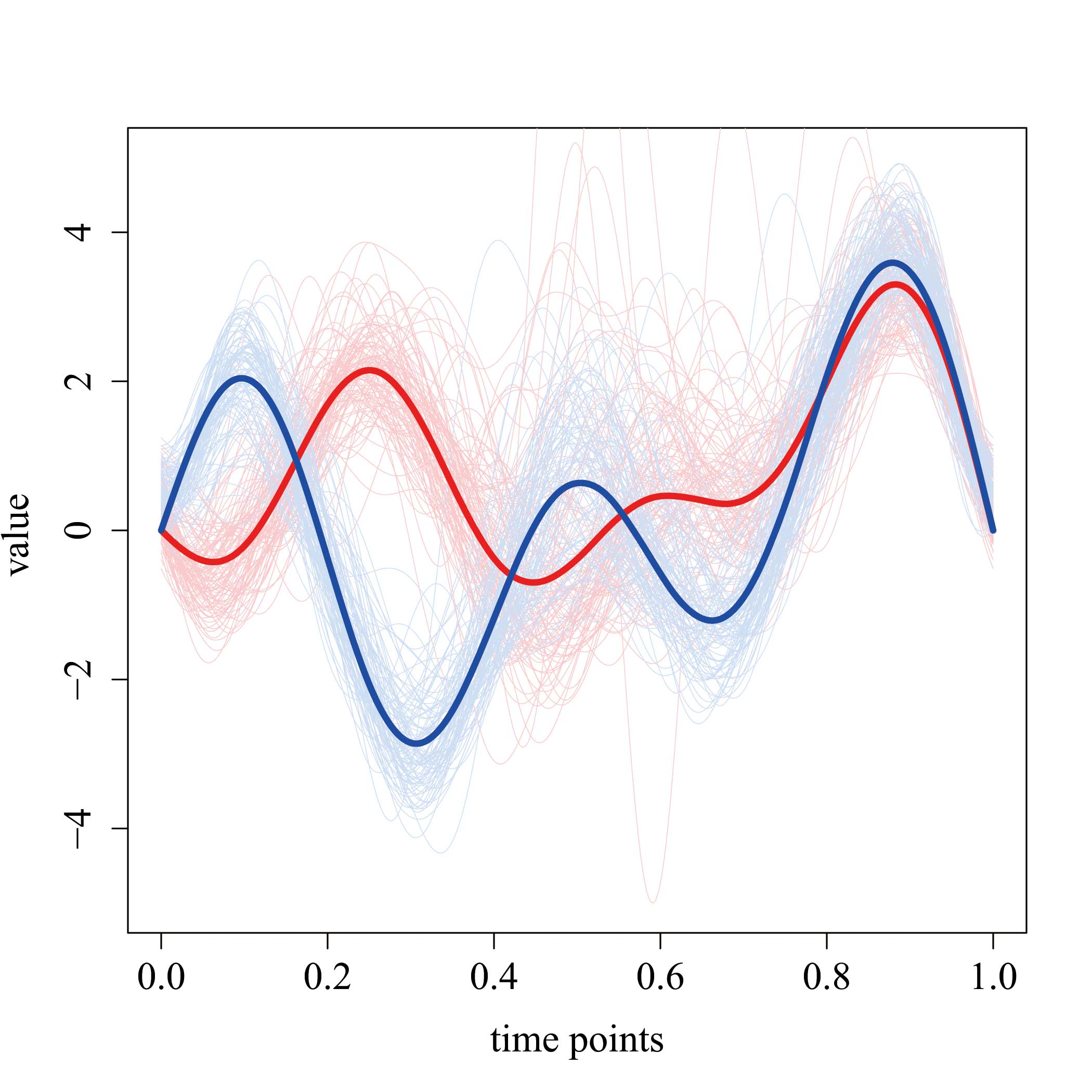} & 	
			\includegraphics[width=4cm]{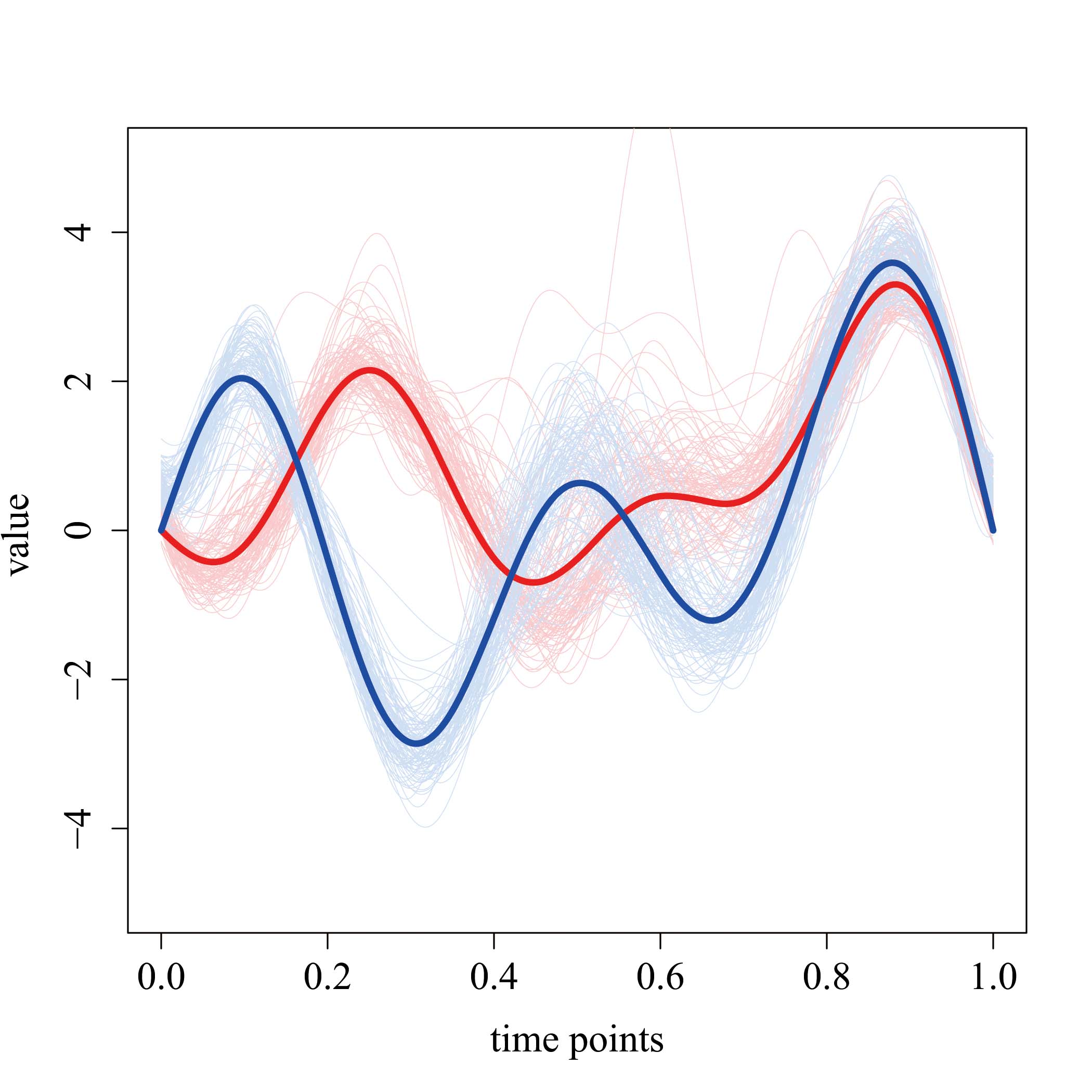} & 	
			\includegraphics[width=4cm]{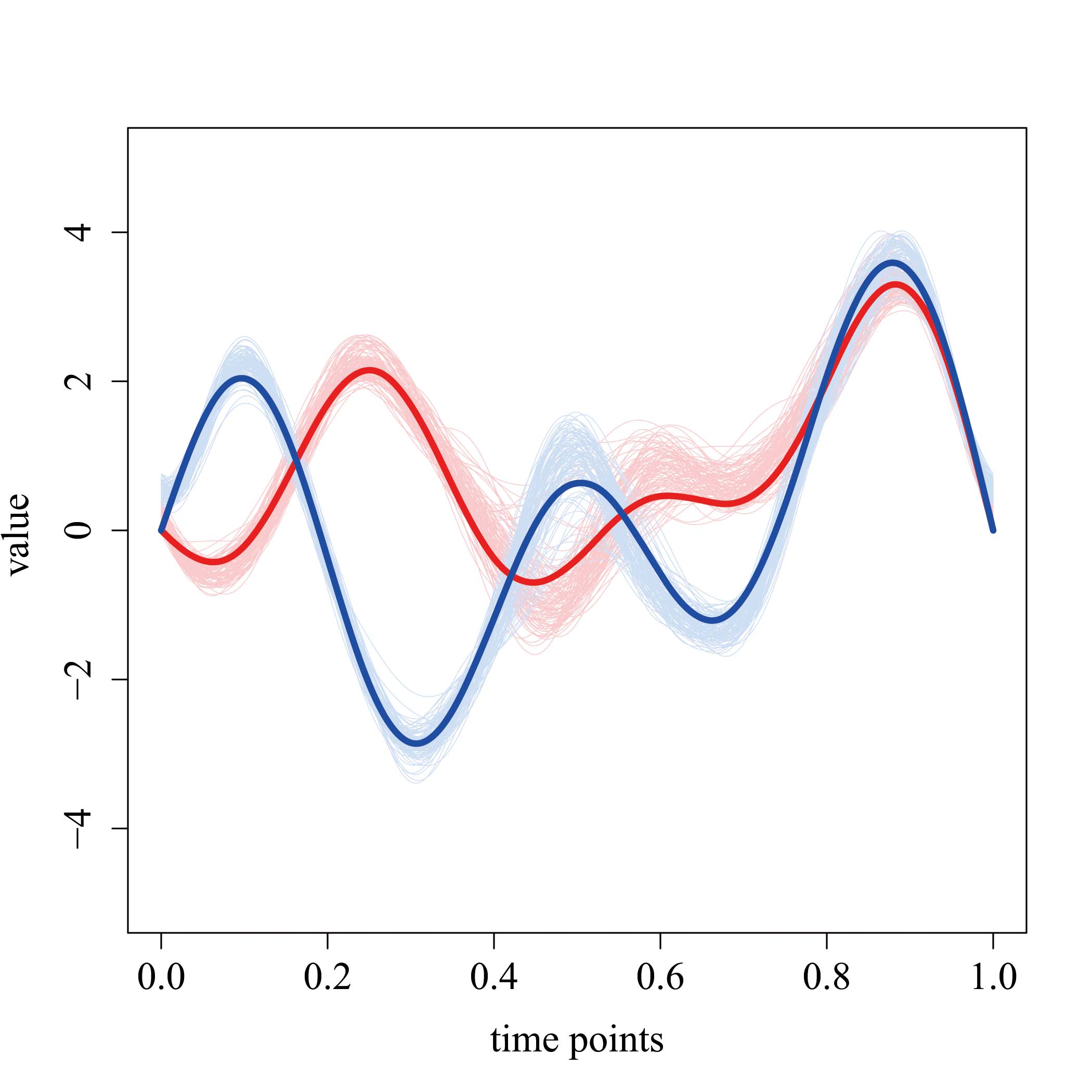}\\\\
			$(n,N_{tp})=(50,10)$ & $(n,N_{tp})=(100,10)$ & $(n,N_{tp})=(200,10)$ & $(n,N_{tp})=(1000,10)$\\
			\includegraphics[width=4cm]{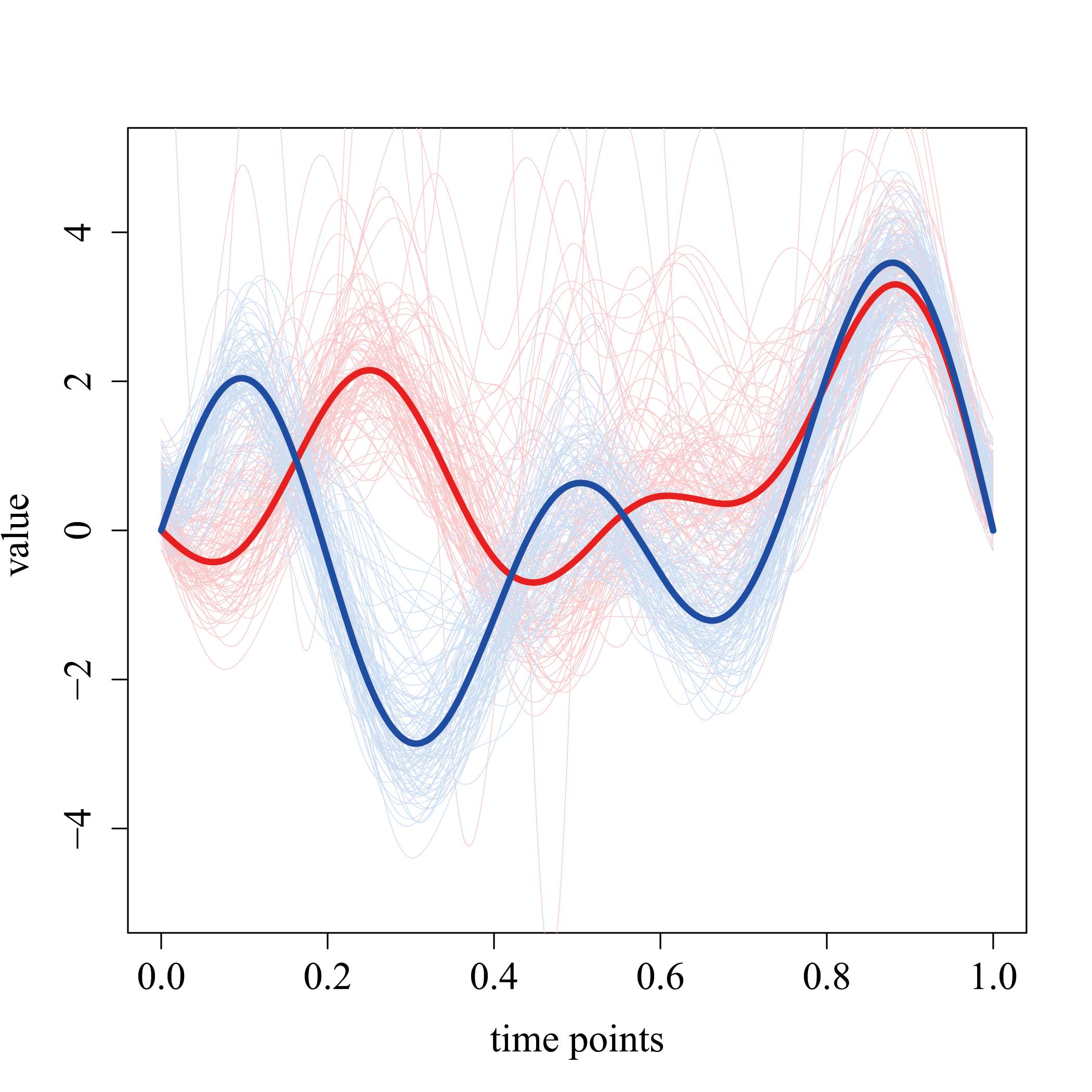} & 	
			\includegraphics[width=4cm]{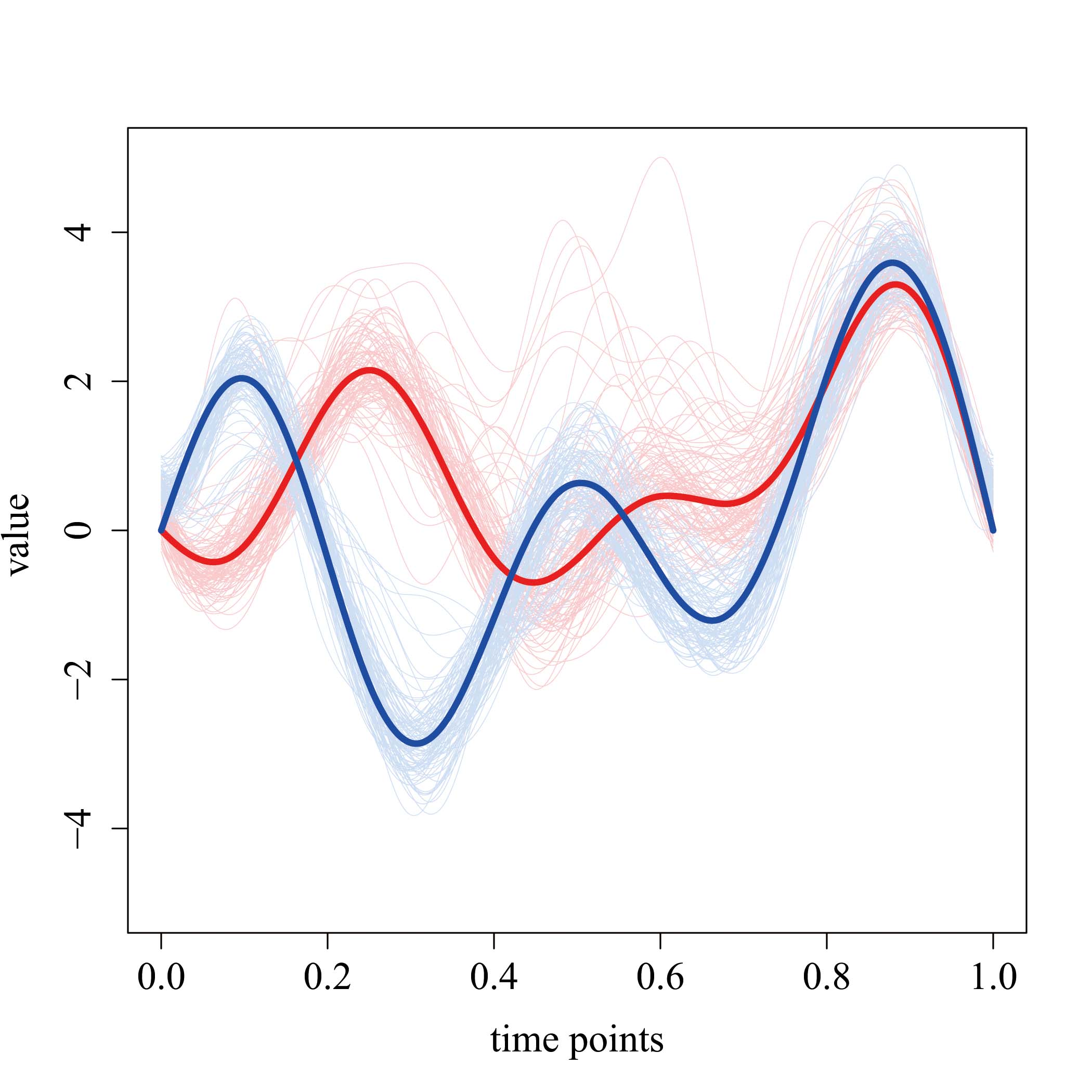} & 	
			\includegraphics[width=4cm]{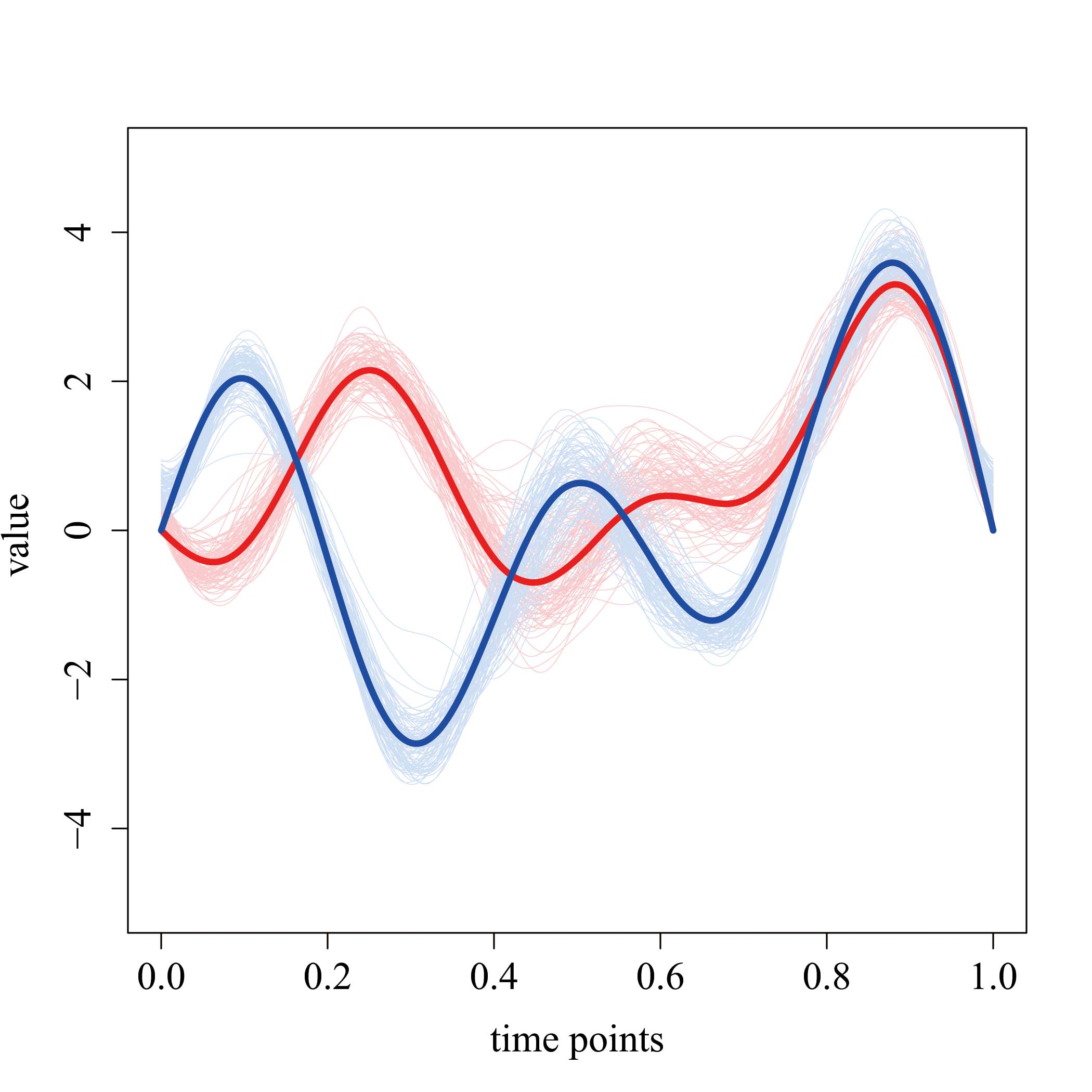} & \includegraphics[width=4cm]{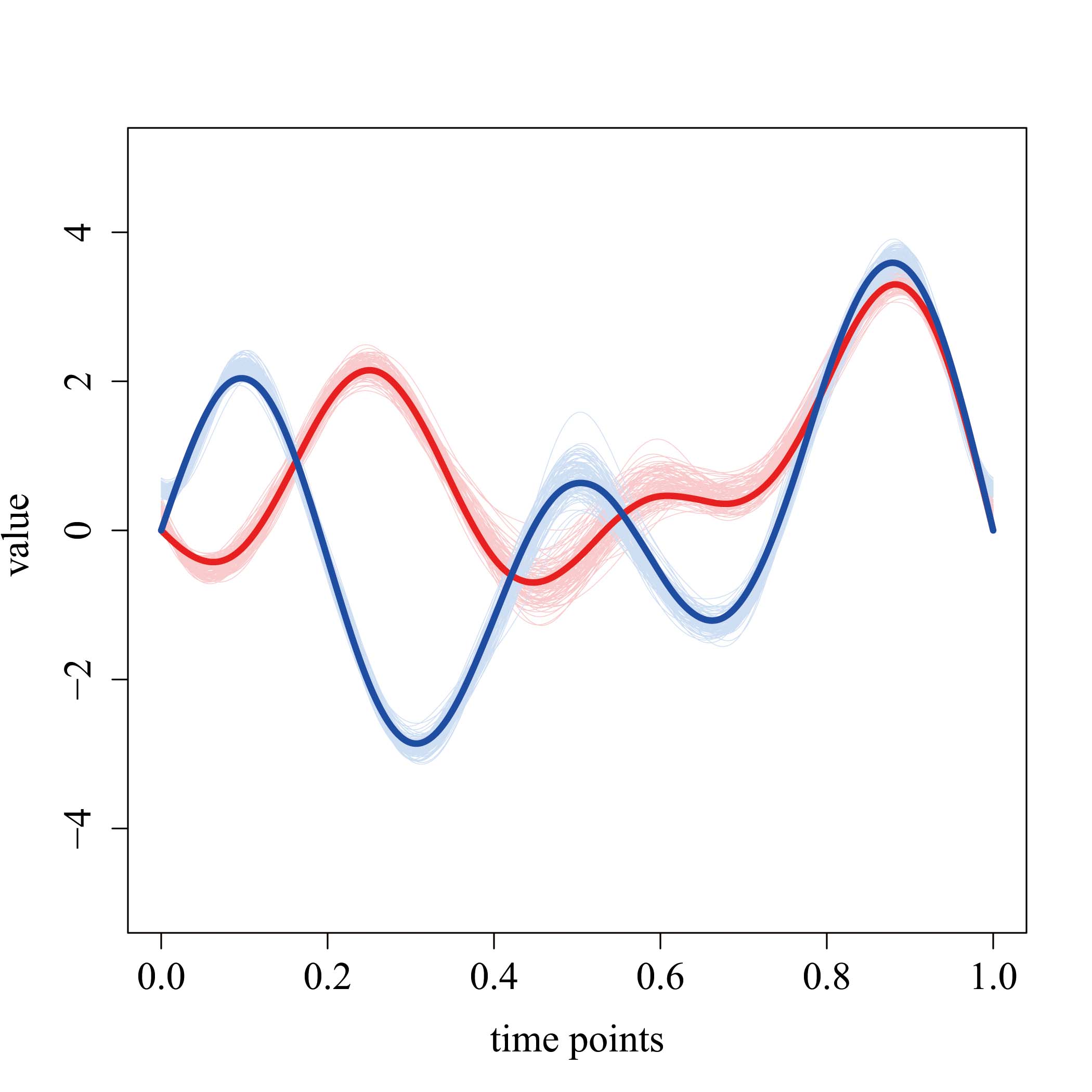}
		\end{tabular}
		\caption{
        Estimated cluster centers (light lines) by FKM-f across 100 replications and optimal cluster centers (bold lines) at the error variance $\sigma=0.1$; 
        panels representing sample sizes $n=50$, 100, 200, and 1000 (from left to right), and the expected number of time points $N_{tp}=3$, 5, and 10 (from top to bottom); clusters in red and blue.
        }
		\label{SI-fig:simulation_convergence_sig=0.1}
	\end{center}
\end{figure}

\clearpage

\subsection{Cluster centers estimated by FKM with B-spline basis functions in the artificial experiment}
\label{SI-subsec:cluster-centers-Bspline}

\begin{figure}[!h]
	\begin{center}
		\renewcommand{\arraystretch}{1}
		\vspace{0.5cm}
		\begin{tabular}{cccc}
			$(n,N_{tp})=(50,3)$ & $(n,N_{tp})=(100,3)$ & $(n,N_{tp})=(200,3)$ & $(n,N_{tp})=(1000,3)$\\
			\includegraphics[width=4cm]
			{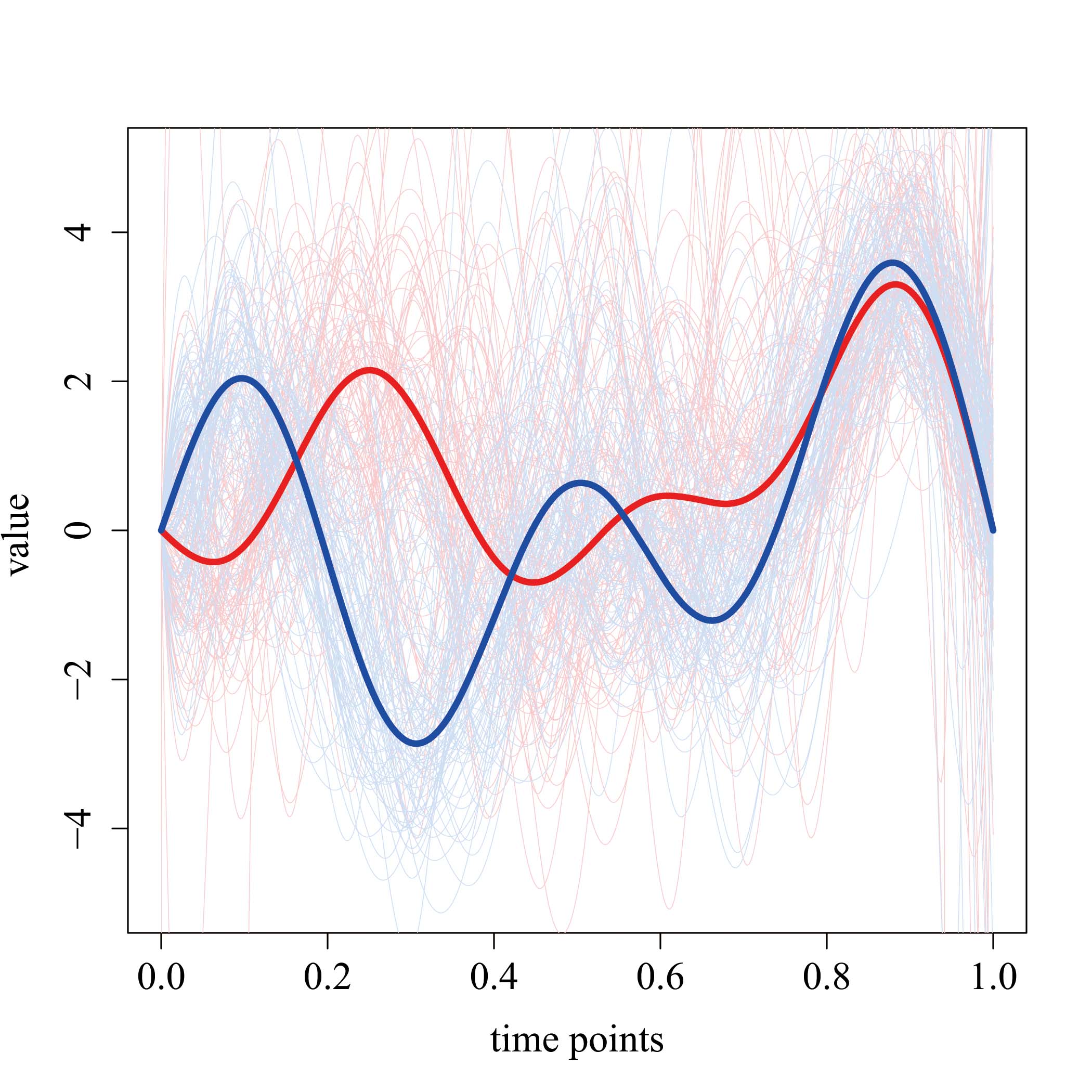} & 	
			\includegraphics[width=4cm]
			{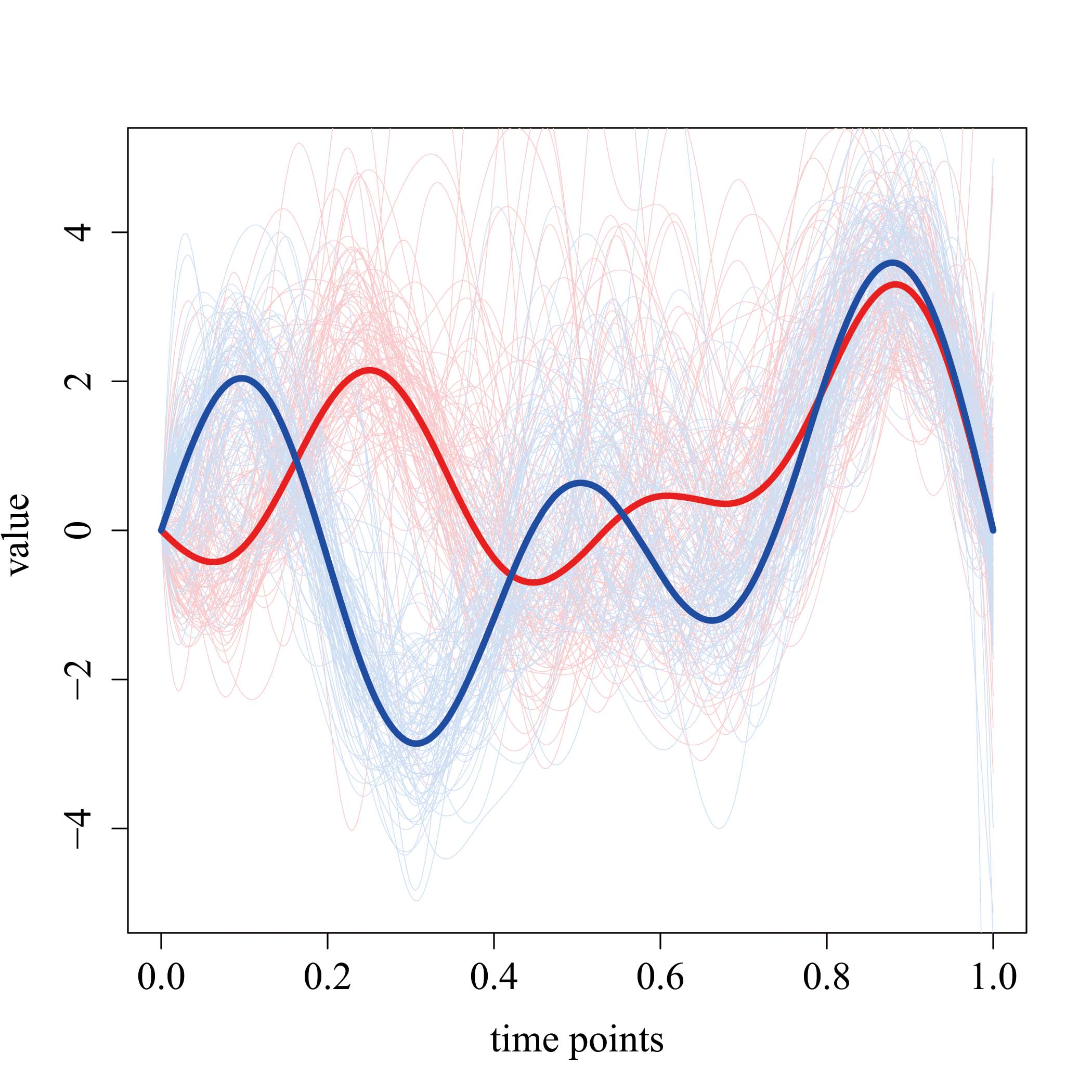} & 	
			\includegraphics[width=4cm]
			{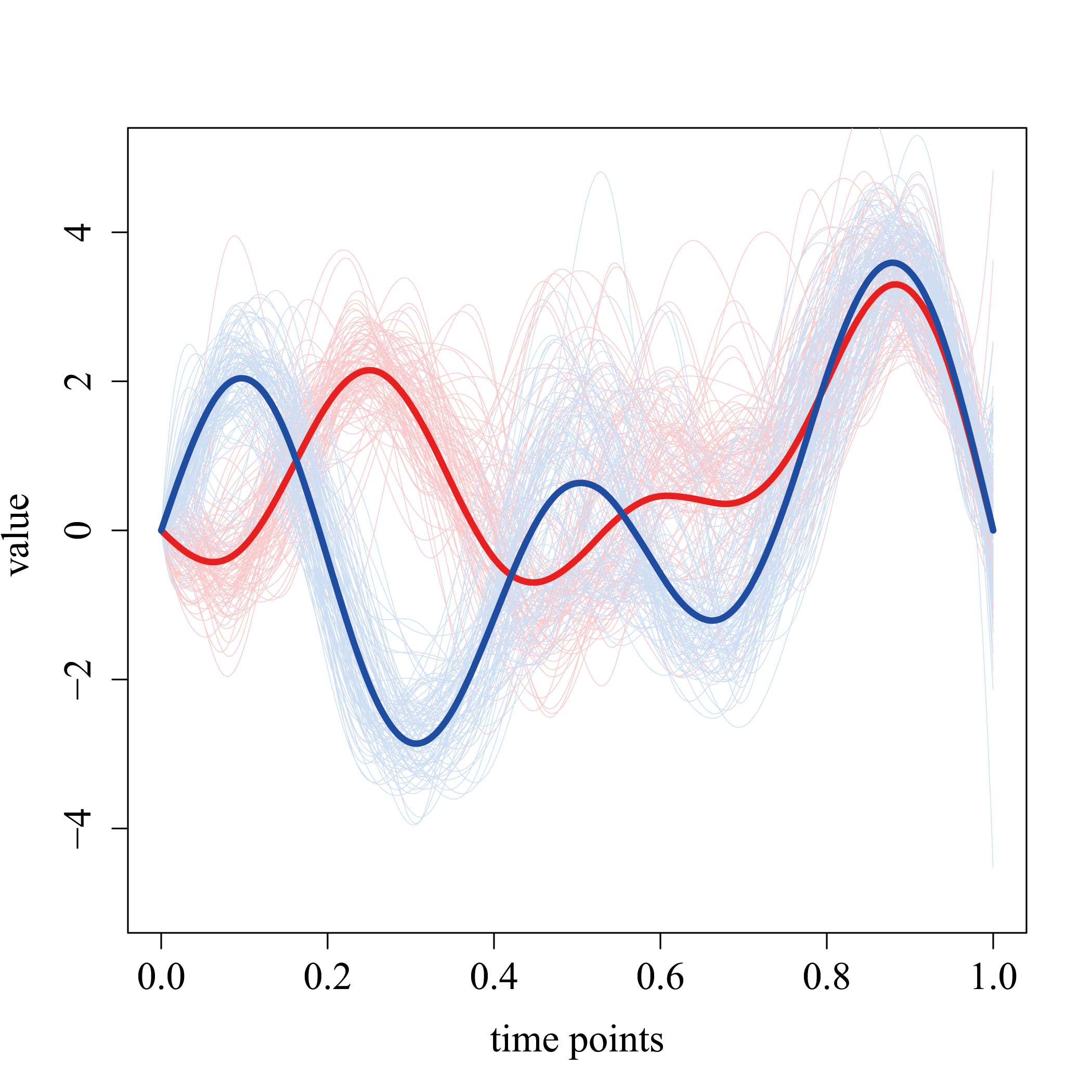} & 	
			\includegraphics[width=4cm]
			{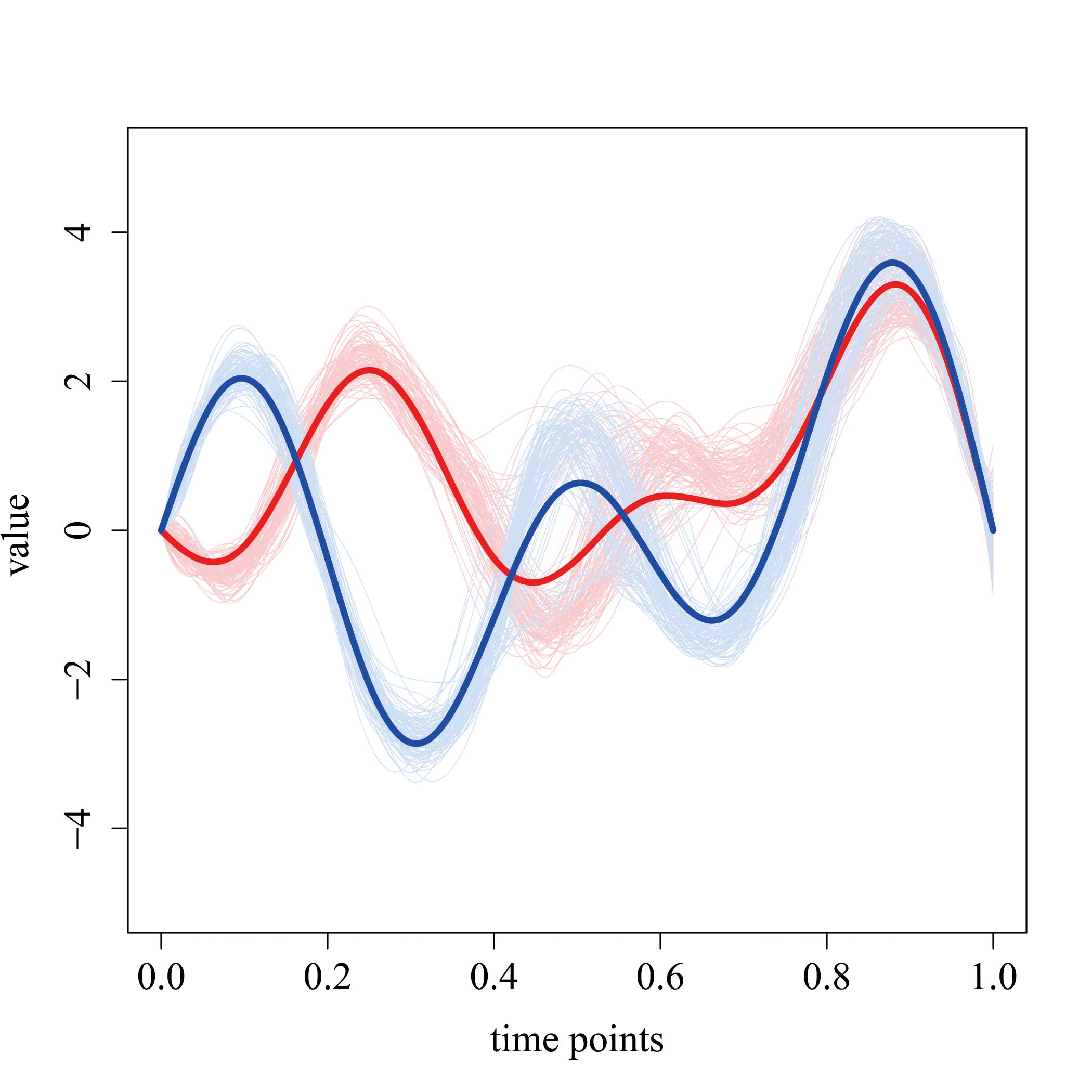}\\\\
			$(n,N_{tp})=(50,5)$ & $(n,N_{tp})=(100,5)$ & $(n,N_{tp})=(200,5)$ & $(n,N_{tp})=(1000,5)$\\
			\includegraphics[width=4cm]
			{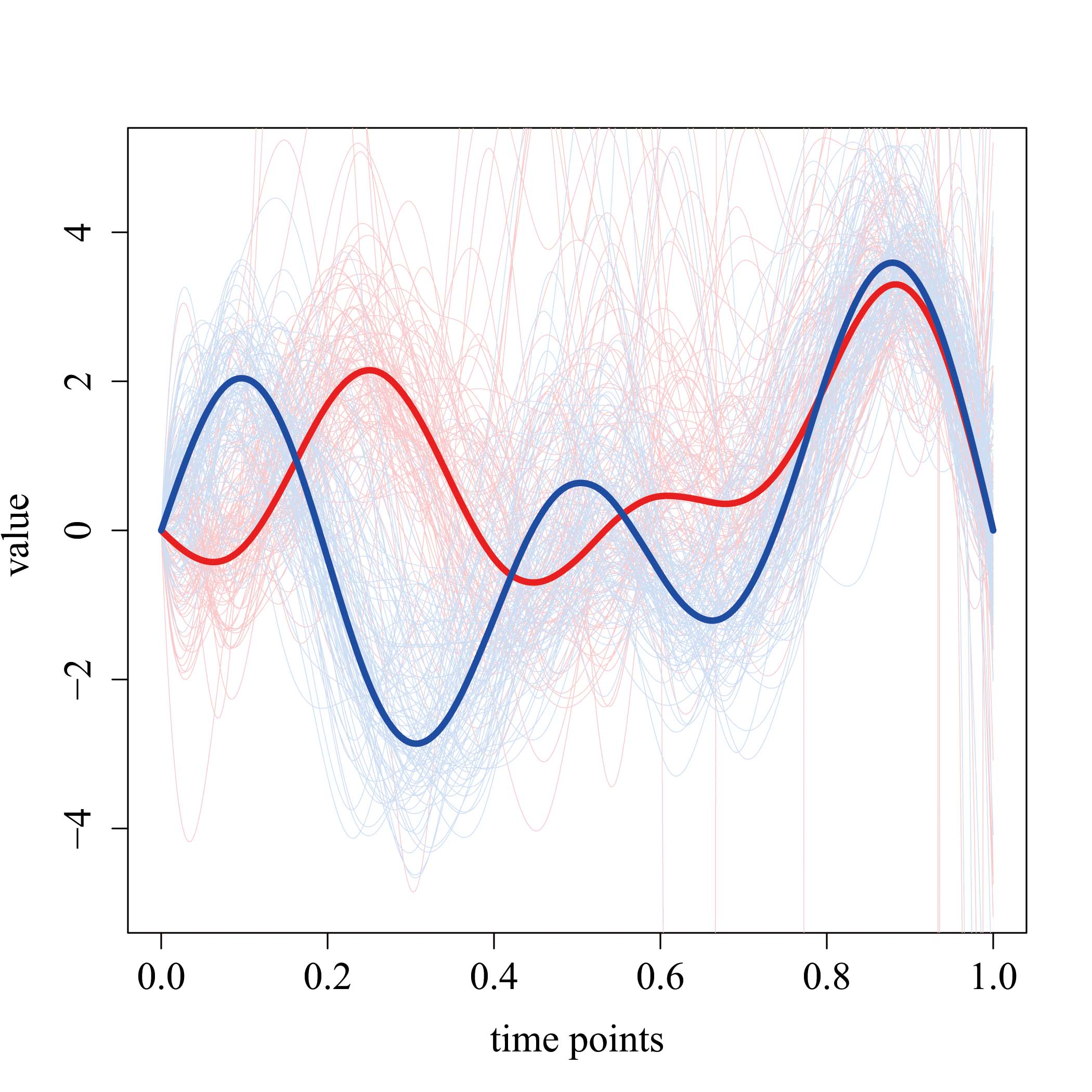} & 	
			\includegraphics[width=4cm]
			{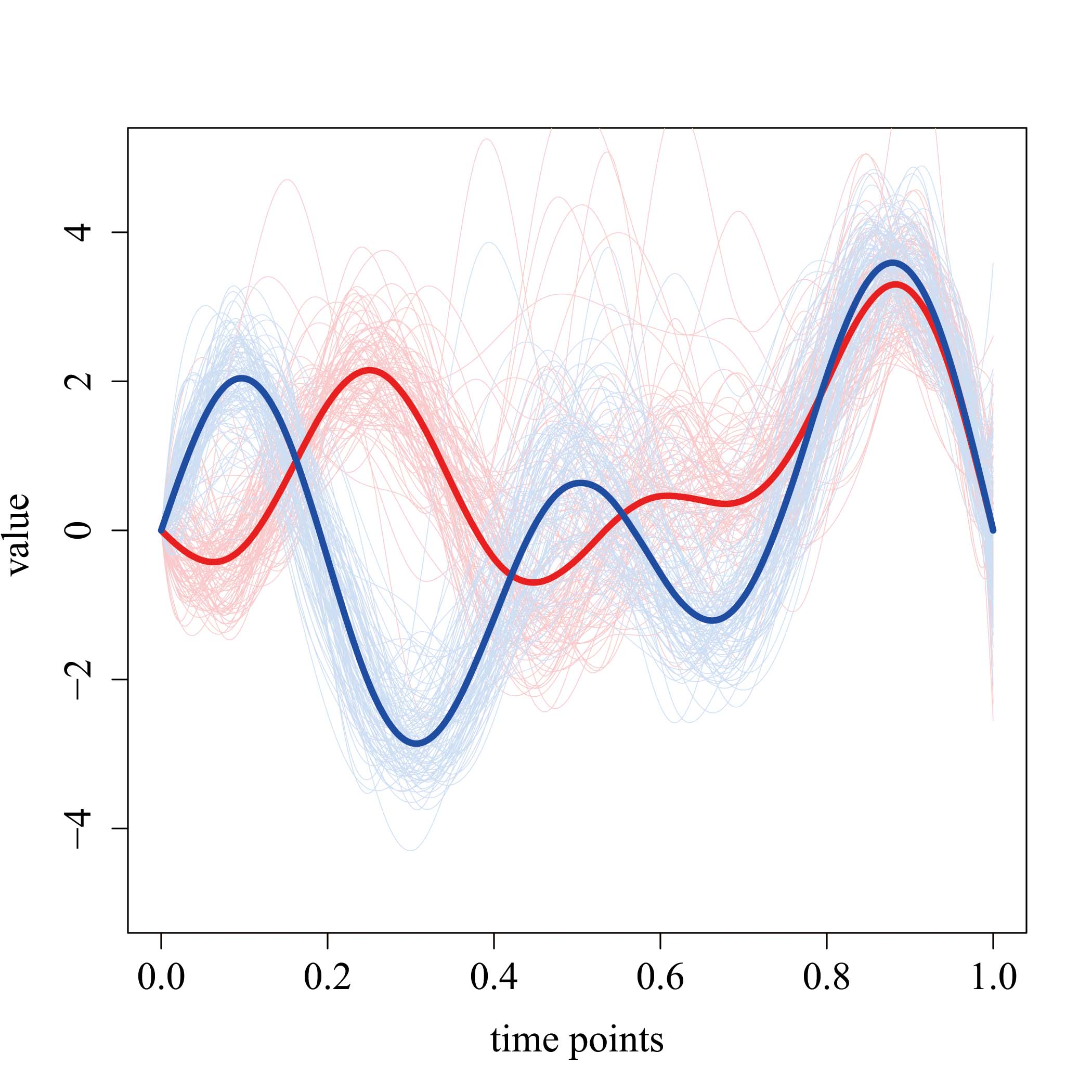} & 	
			\includegraphics[width=4cm]
			{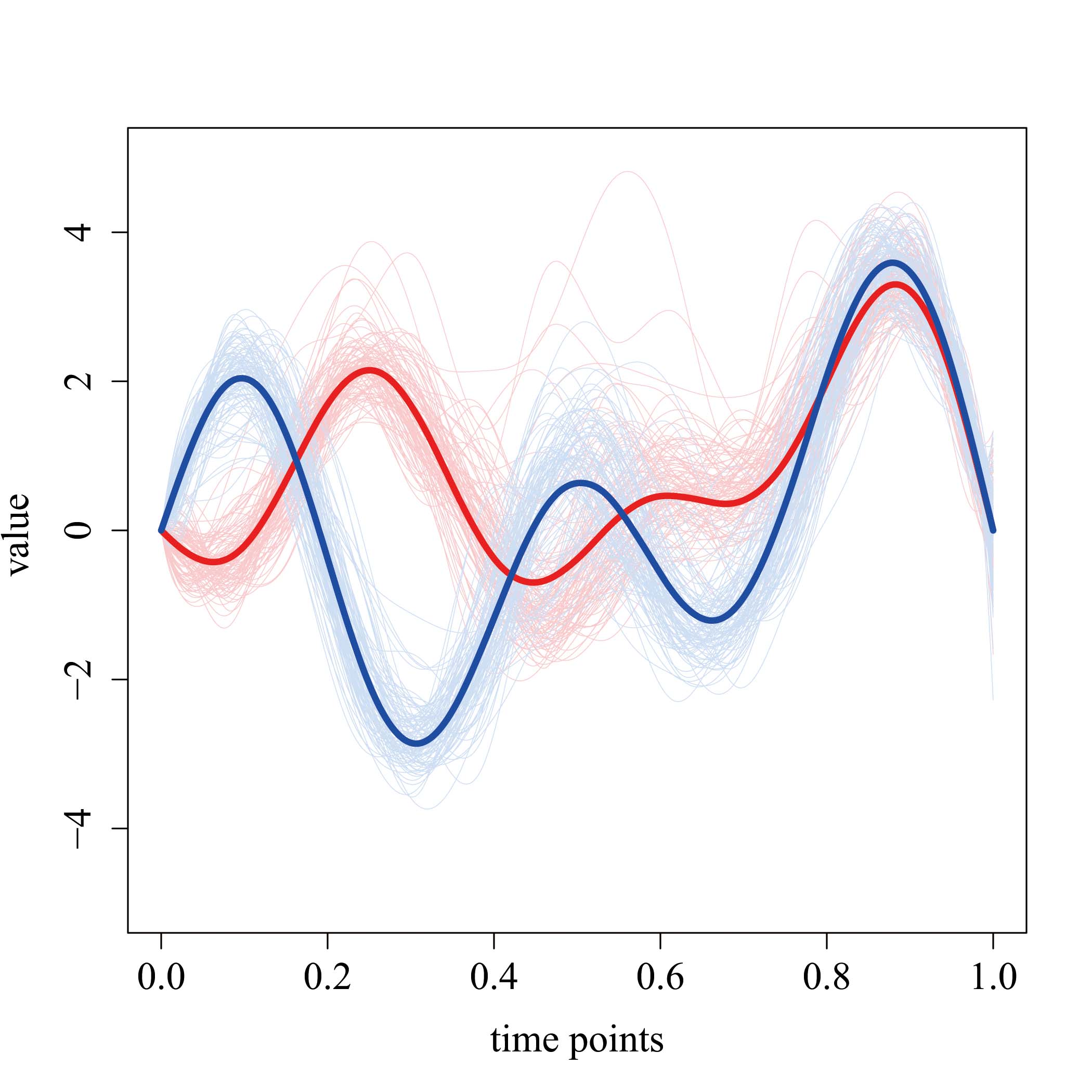} & 	
			\includegraphics[width=4cm]
			{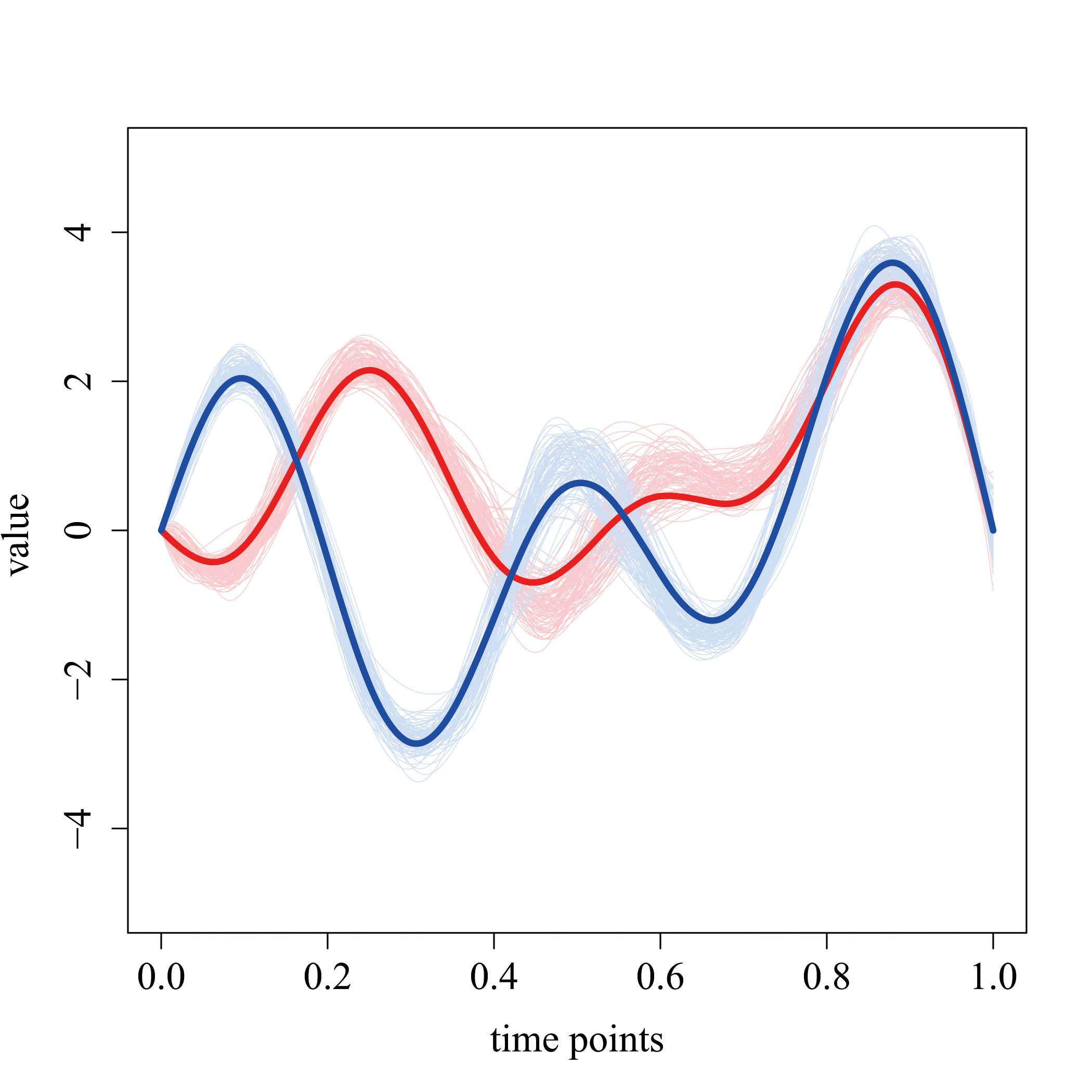}\\\\
			$(n,N_{tp})=(50,10)$ & $(n,N_{tp})=(100,10)$ & $(n,N_{tp})=(200,10)$ & $(n,N_{tp})=(1000,10)$\\
			\includegraphics[width=4cm]
			{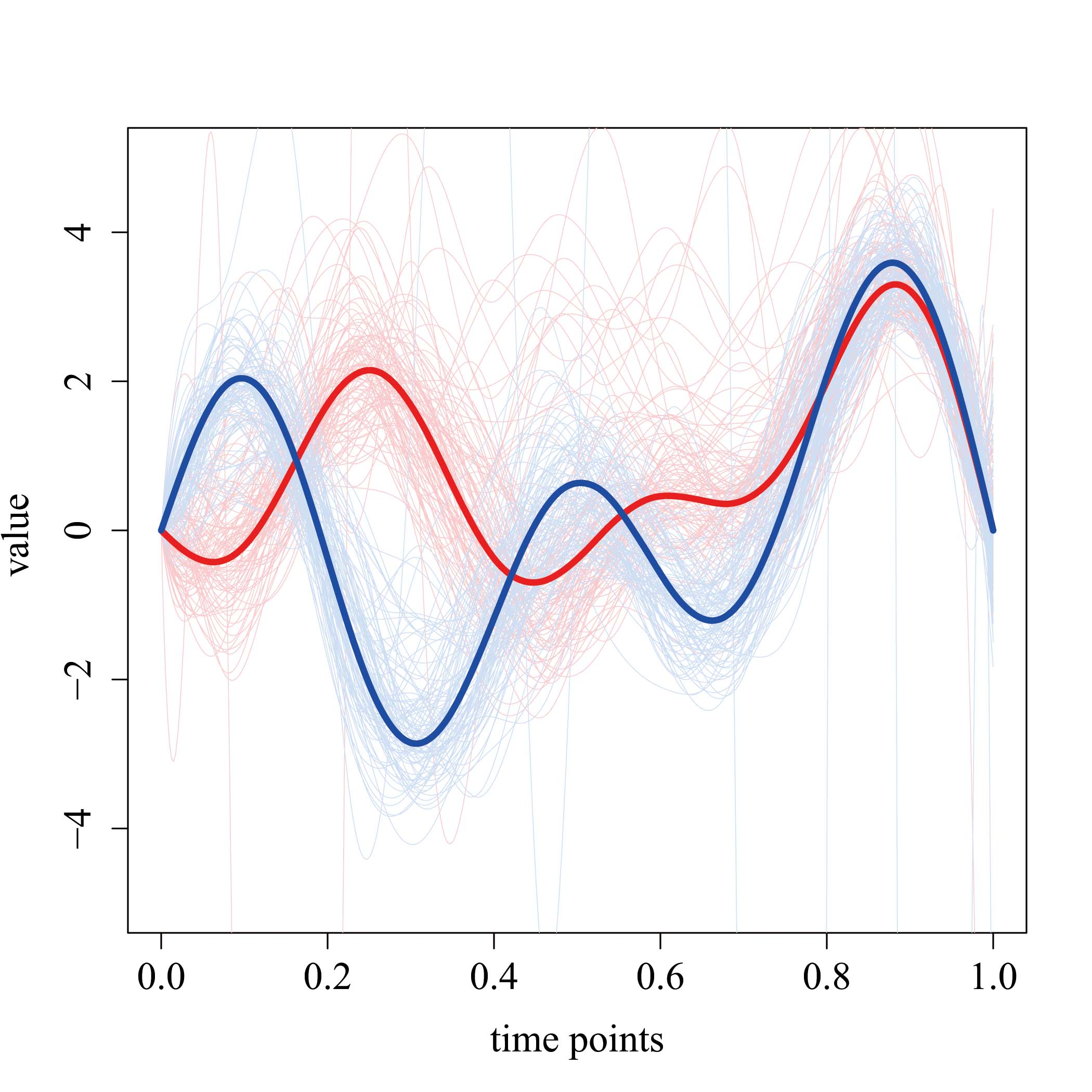} & 	
			\includegraphics[width=4cm]
			{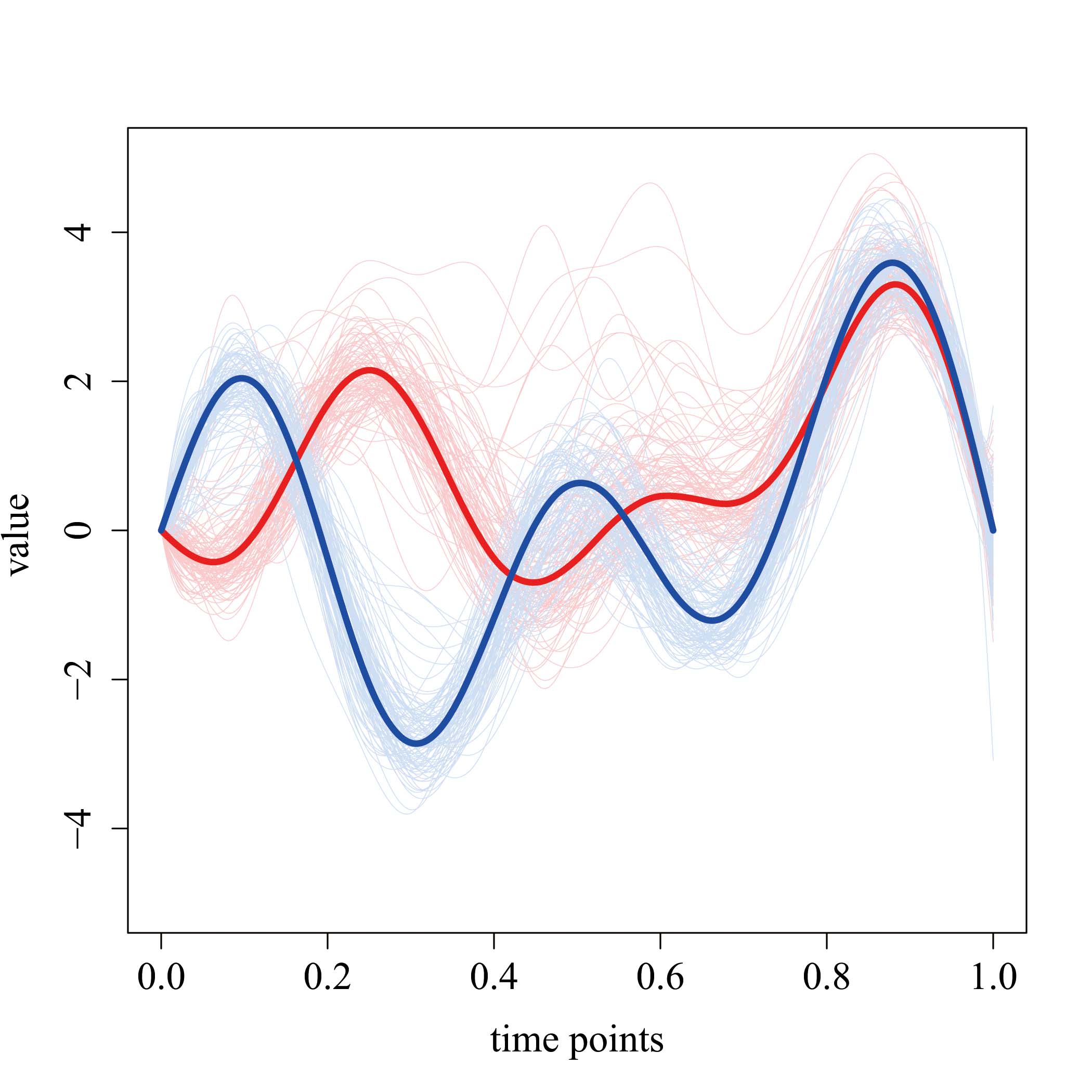} & 	
			\includegraphics[width=4cm]
			{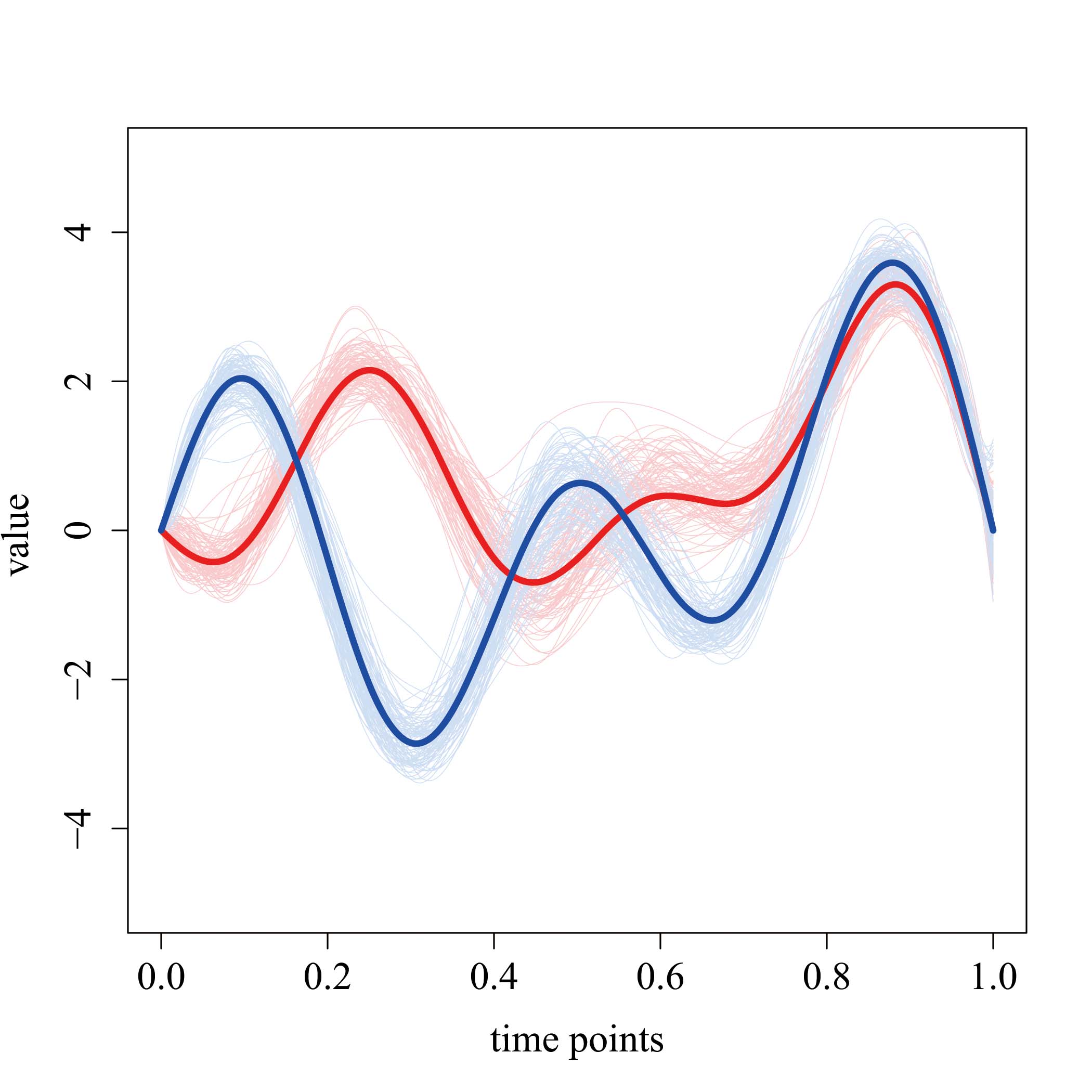} & 	
			\includegraphics[width=4cm]
			{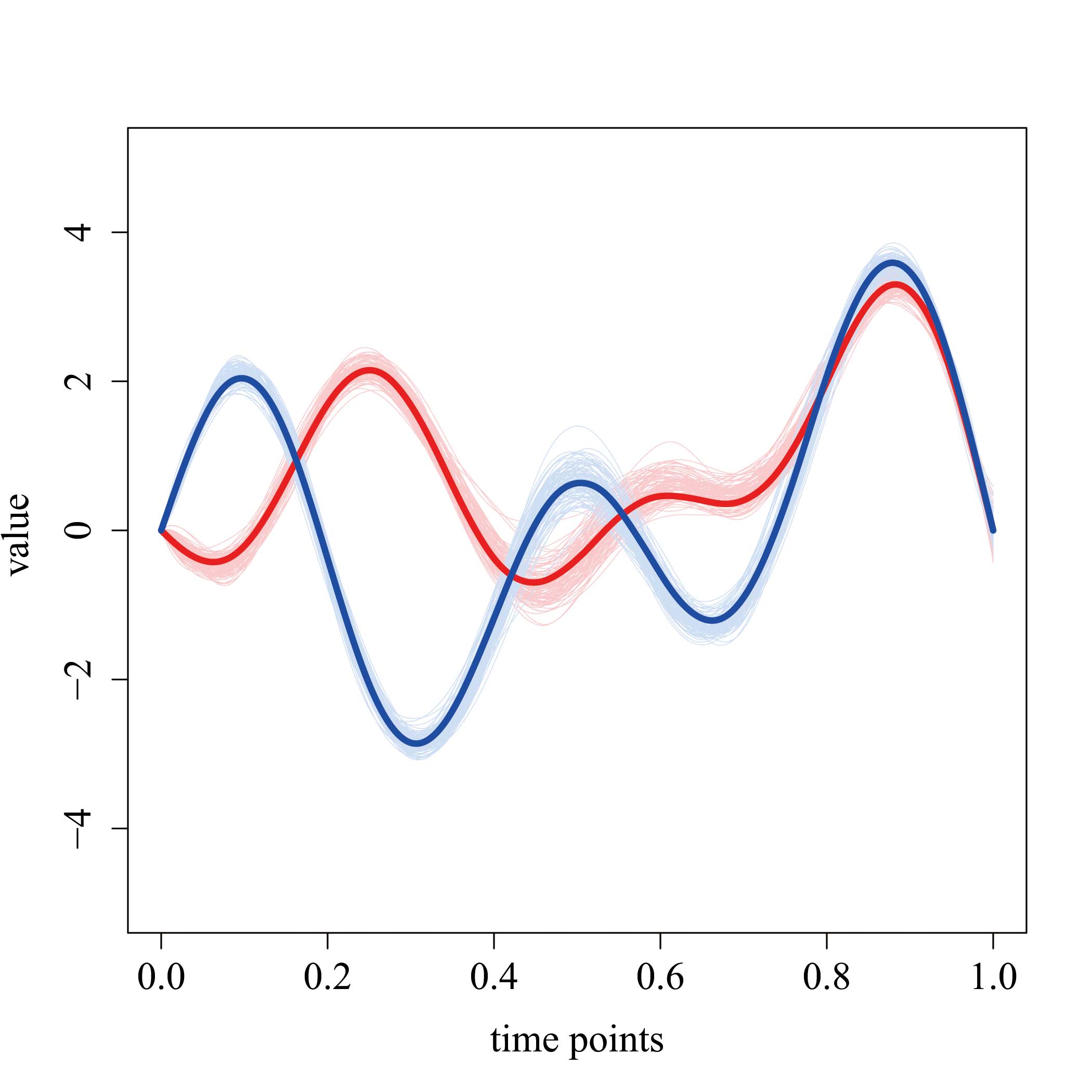}
		\end{tabular}
		\caption{
                Estimated cluster centers (light lines) by FKM-b across 100 replications and optimal cluster centers (bold lines) at the error variance $\sigma=0.1$; 
                panels representing sample sizes $n=50$, 100, 200, and 1000 (from left to right), and the expected number of time points $N_{tp}=3$, 5, and 10 (from top to bottom); 
                clusters in red and blue.
  }
		\label{SI-fig:simulation_convergence_sig=0.1}
	\end{center}
\end{figure}

\clearpage

\subsection{Adjusted Rand indices in the artificial experiment}
\label{SI-subsec:ARI}

\begin{table}[!h]
	\caption{Adjusted Rand indices in percentages averaged over 100 replications under the combinations of conditions $(N_{tp},\sigma,n)$.}
	\label{SI-tab:simulation-results}
  \centering
    \begin{tabular}{cccccccccccccc}
    \toprule
          &       &       & \multicolumn{3}{c}{$\sigma$=0.1} & \multicolumn{4}{c}{$\sigma$=1.0} & \multicolumn{4}{c}{$\sigma$=2.0} \\
\cmidrule{3-14}    $N_{tp}$  & Methods & \multicolumn{4}{c}{$n$}         & \multicolumn{4}{c}{$n$}         & \multicolumn{4}{c}{$n$} \\
          &       & 50    & 100   & 200   & 400   & 50    & 100   & 200   & 400   & 50    & 100   & 200   & 400 \\
    \midrule
    3     & FKM-f & 12.4  & 15.7  & 20.2  & 22.7  & 8.8   & 11.1  & 13.4  & 15.6  & 3.4   & 3.4   & 4.4   & 7.2 \\
          & FKM-b & 13.6  & 15.7  & 20.7  & 23.3  & 9.4   & 11.8  & 13.6  & 15.8  & 3.2   & 4.7   & 4.7   & 6.6 \\
          & FCM   & 12.8  & 16.2  & 22.0  & 27.1  & 9.3   & 10.6  & 13.8  & 16.3  & 3.7   & 4.2   & 5.3   & 6.6 \\
          & distclust & 3.9   & 5.8   & 7.4   & 9.1   & 3.4   & 3.2   & 3.5   & 5.9   & 1.7   & 1.7   & 1.3   & 2.1 \\
    \midrule
    5     & FKM-f & 20.9  & 35.7  & 39.5  & 42.4  & 12.7  & 22.3  & 27.8  & 31.4  & 6.4   & 9.5   & 12.4  & 15.0 \\
          & FKM-b & 21.4  & 35.6  & 40.2  & 42.4  & 14.0  & 23.7  & 26.7  & 31.6  & 6.6   & 9.9   & 12.6  & 14.5 \\
          & FCM   & 27.4  & 42.7  & 48.1  & 51.6  & 12.4  & 26.5  & 33.6  & 38.6  & 6.6   & 10.2  & 13.0  & 16.7 \\
          & distclust & 10.6  & 15.2  & 20.0  & 19.5  & 8.2   & 10.5  & 13.8  & 14.2  & 2.6   & 3.1   & 4.7   & 5.9 \\
    \midrule
    10    & FKM-f & 48.4  & 56.1  & 61.5  & 64.3  & 39.1  & 45.2  & 50.6  & 52.6  & 22.3  & 26.0  & 30.6  & 33.0 \\
          & FKM-b & 51.4  & 55.8  & 61.5  & 64.1  & 39.4  & 45.2  & 51.9  & 53.7  & 23.0  & 26.3  & 31.1  & 32.2 \\
          & FCM   & 68.2  & 76.5  & 78.9  & 81.0  & 50.2  & 58.3  & 65.2  & 67.7  & 21.9  & 30.0  & 37.6  & 43.0 \\
          & distclust & 36.6  & 42.3  & 42.6  & 45.0  & 26.1  & 30.3  & 36.6  & 36.0  & 13.6  & 15.1  & 17.0  & 20.3 \\
    \bottomrule
    \end{tabular}%
\end{table}%

\clearpage

\subsection{Median time for execution}
\label{SI-subsec:implementation-median-time}

\begin{figure}[!h]
	\begin{center}
	\renewcommand{\arraystretch}{1}
	 \vspace{0.5cm}
		\begin{tabular}{c}
		 \includegraphics[width=16cm]{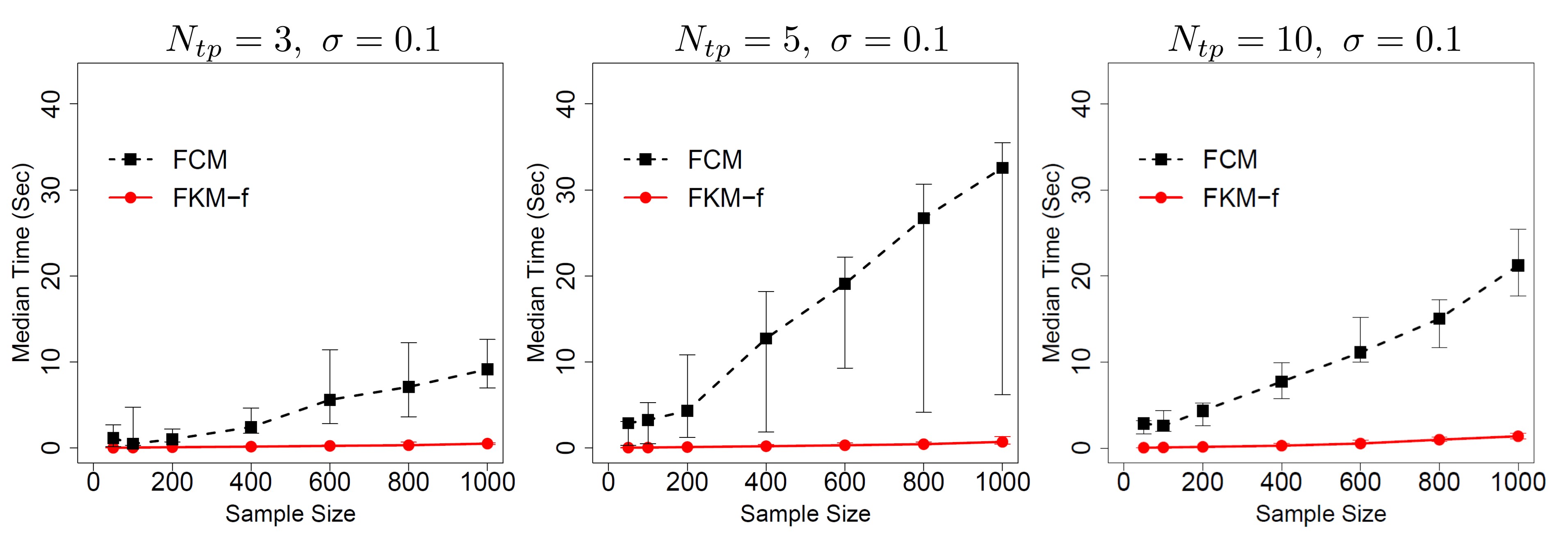}\\\\
		 \includegraphics[width=16cm]{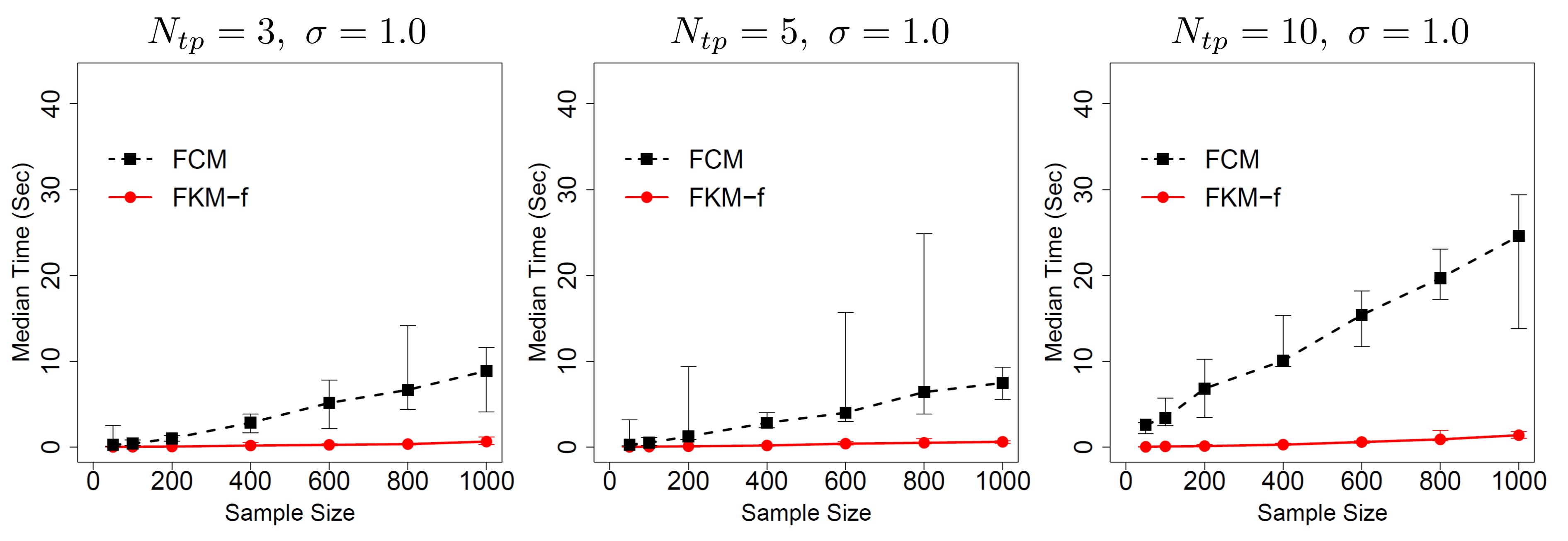}\\\\
		 \includegraphics[width=16cm]{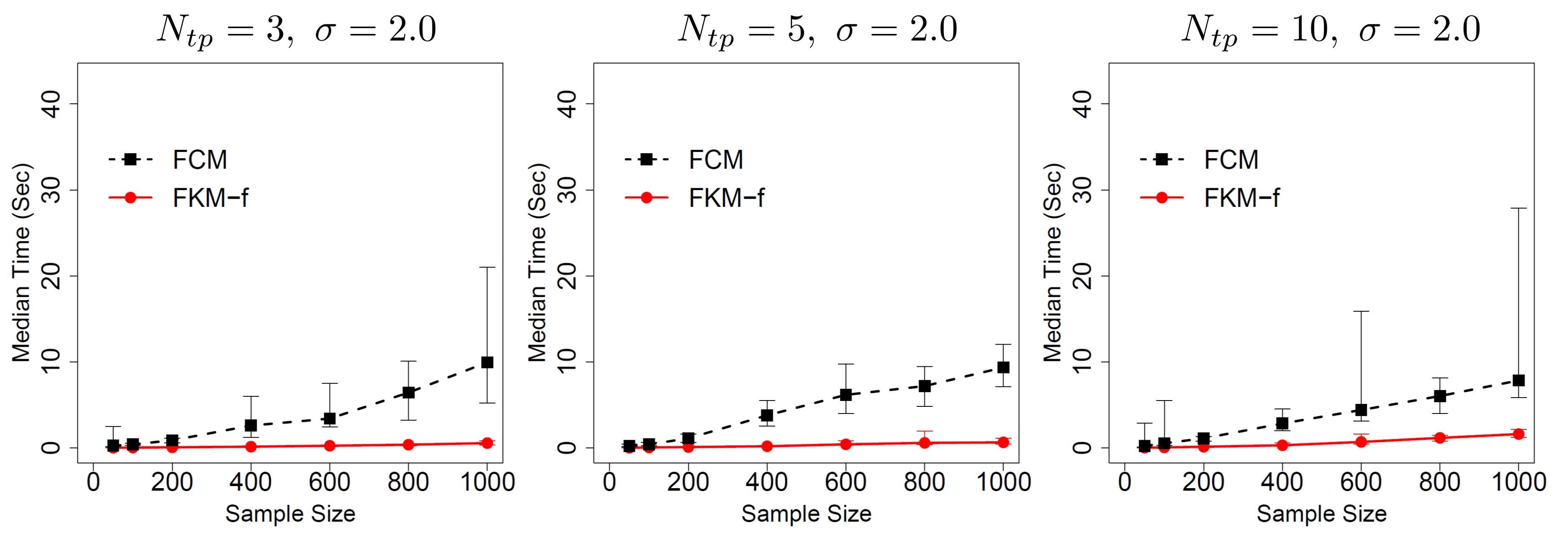}\\\\
		\end{tabular}
		\caption{
            Median analysis time (in seconds) with error bars showing the range (min-max) for each sample size at error variance $\sigma=1.0$ (black dashed line: FCM, red solid line: FKM-f);
            panels representing results for expected time points $N_{tp}=3$, 5, and 10 (from left to right), and the error variances $\sigma=0.1$, 1.0, and 2.0 (from top to bottom).
  }
	 \label{fig:implementation-median-time}
	\end{center}
\end{figure}



\end{document}